\documentclass[12pt]{article}
\usepackage{preamble}
\usepackage{mdframed}
\allowdisplaybreaks
\pdfoutput=1

\def\AND{\mathsf{AND}}
\def\OR{\mathsf{OR}}

\newcommand{\red}[1]{{\color{red}{#1}}}

\newcommand{\jnote}[1]{\footnote{{\bf \color{green}Jarek}: {#1}}}

\newtheorem*{rep@theorem}{\rep@title}
\newcommand{\newreptheorem}[2]{
\newenvironment{rep#1}[1]{
 \def\rep@title{#2 \ref{##1}}
 \begin{rep@theorem}\itshape}
 {\end{rep@theorem}}}
\makeatother
\theoremstyle{plain}

\def\inner{\mathsf{in}}
\def\outer{\mathsf{out}}

\begin{document}

\title{Fourier growth of structured $\mathbb{F}_2$-polynomials\\ and applications}

\author{
Jaros\l{}aw B\l{}asiok\\
Columbia University\\
jb4451@columbia.edu
\and
Peter Ivanov\\
Northeastern University\\
ivanov.p@northeastern.edu
\and
Yaonan Jin
\\
Columbia University\\
yj2552@columbia.edu
\and
Chin Ho Lee\\
Harvard University\\
chlee@seas.harvard.edu
\and
Rocco A. Servedio\\
Columbia University\\
rocco@cs.columbia.edu
\and
Emanuele Viola\\
Northeastern University\\
viola@ccs.neu.edu
}

\maketitle

\abstract{
We analyze the \emph{Fourier growth}, i.e.~the $L_1$ Fourier weight at level $k$ (denoted $L_{1,k}$),  of various well-studied classes of ``structured'' $\F_2$-polynomials.  This study is motivated by applications in pseudorandomness, in particular recent results and conjectures due to \cite{CHHL,CHLT,CGLSS} which show that upper bounds on Fourier growth (even at level $k=2$) give unconditional pseudorandom generators. 

Our main results on Fourier growth are as follows:

\begin{itemize}

  \item We show that any symmetric degree-$d$ $\F_2$-polynomial $p$ has $L_{1,k}(p) \leq \Pr[p=1] \cdot O(d)^k$, and this is tight for any constant $k$. This quadratically strengthens an earlier bound that was implicit in \cite{RSV13}.
\ignore{		\red{We also establish lower bounds showing that this upper bound is best possible for constant $k$.\jnote{I find this statement a bit confusing --- somehow I read it as ``for every $d$ and $k$, and arbitrarily large $n$, there is a symmetric degree $d$ polynomial on $n$ variables with $L_{1,k} \geq d^k$'' --- which it clearly is not.} }
}

\item We show that any read-$\Delta$ degree-$d$ $\F_2$-polynomial $p$ has $L_{1,k}(p) \leq \Pr[p=1] \cdot (k \Delta d)^{O(k)}$.

\item We establish a composition theorem which gives $L_{1,k}$ bounds on disjoint compositions of functions that are closed under restrictions and admit $L_{1,k}$ bounds. 

\end{itemize}

Finally, we apply the above results to obtain new unconditional pseudorandom generators and new correlation bounds for various classes of $\F_2$-polynomials.
}

\thispagestyle{empty}

\newpage

\setcounter{page}{1}


\section{Introduction} \label{sec:intro}

\subsection{Background:  $L_1$ Fourier norms and Fourier growth}

Over the past several decades, Fourier analysis of Boolean functions has emerged as a fundamental tool of great utility  across many different areas within theoretical computer science and mathematics.  Areas of application include (but are not limited to) combinatorics, the theory of random graphs and statistical physics, social choice theory, Gaussian geometry and  the study of metric spaces, cryptography, learning theory, property testing, and many branches of computational complexity such as hardness of approximation, circuit complexity, and pseudorandomness.  The excellent book of O'Donnell \cite{ODonnell:book} provides a broad introduction. In this paper we follow the notation of \cite{ODonnell:book}, and for a Boolean-valued function $f$ on $n$ Boolean variables and $S \subseteq [n]$, we write $\widehat{f}(S)$ to denote the Fourier coefficient of $f$ on $S$.

Given the wide range of different contexts within which the Fourier analysis of Boolean functions has been pursued, it is not surprising that many different quantitative parameters of Boolean functions have been analyzed in the literature. In this work we are chiefly interested in the \emph{$L_1$ Fourier norm at level $k$}:

\begin{definition}[$L_1$ Fourier norm at level $k$]
	The \emph{$L_1$ Fourier norm of a function $f\colon\pmone^n \to \zo$ at level $k$} is the quantity
	\begin{align*}
		L_{1,k}(f) := \sum_{S \subseteq [n] : \abs{S}=k} \abs{\hf(S)} .
	\end{align*}
	For a function class $\calF$, we write $L_{1,k}(\calF)$ to denote $\sup_{f \in \calF} L_{1,k}(f)$.
\end{definition}

As we explain below, strong motivation for studying the $L_1$ Fourier norm at level $k$ (even for specific small values of $k$ such as $k=2$) is given by exciting recent results in unconditional \emph{pseudorandomness}.
More generally, the notion of \emph{Fourier growth} is a convenient way of capturing the $L_1$ Fourier norm at level $k$ for every $k$:

\begin{definition}[Fourier growth]
	A function class $\calF \subseteq \{ f\colon\pmone^n \to \zo \}$ has \emph{Fourier growth $L_1(a,b)$} if there exist constants $a$ and $b$ such that $L_{1,k}(\calF) \le a \cdot b^k$ for every $k$.
\end{definition}

The notion of Fourier growth was explicitly introduced by  Reingold, Steinke, and Vadhan in \cite{RSV13} for the purpose of constructing pseudorandom generators for space-bounded computation (though we note that the Fourier growth of DNF formulas was already analyzed in \cite{Man95}, motivated by applications in learning theory). In recent years there has been a surge of research interest in understanding the Fourier growth of different types of functions \cite{Tal17,GSTW16,CHRT,Lee,GRZ,Tal,SSW,GRT}.  One strand of motivation for this study has come from the study of quantum computing;
in particular, bounds on the Fourier growth of $\AC^0$~\cite{Tal17} were used in the breakthrough result of Raz and Tal~\cite{RazT19} which gave an oracle separation between the classes $\BQP$ and $\PH$.
More recently, in order to achieve an optimal separation between quantum and randomized query complexity, several researchers~\cite{Tal,BS,SSW} have studied the Fourier growth of decision trees, with the recent work of~\cite{SSW} obtaining optimal bounds.  Analyzing the Fourier growth of other classes of functions has also led to separations between quantum and classical computation in other settings~\cite{GRT,GRZ,GirishTW21}.

Our chief interest in the current paper arises from a different line of work which has established powerful applications of Fourier growth bounds in pseudorandomness.
We describe the relevant background, which motivates a new conjecture that we propose on Fourier growth, in the next subsection.

\subsection{Motivation for this work:  Fourier growth, pseudorandomness, $\F_2$-polynomials, and the \cite{CHLT} conjecture}

\paragraph{Pseudorandom generators from Fourier growth bounds.}
Constructing explicit, unconditional pseudorandom generators (PRGs) for various classes of Boolean functions is an important goal in complexity theory.  
In the recent work \cite{CHHL}, Chattopadhyay, Hatami, Hosseini, and Lovett introduced a novel framework for the design of such PRGs.
Their approach provides an explicit pseudorandom generator for any class of functions that
is closed under restrictions and has bounded Fourier growth:

\begin{theorem} [PRGs from Fourier growth: Theorem~23 of \cite{CHHL}] \label{thm:CHHL}
	Let ${\cal F}$ be a family of $n$-variable Boolean functions that is closed under restrictions and has Fourier growth $L_1(a,b)$.  Then there is an explicit pseudorandom generator that $\eps$-fools ${\cal F}$ with seed length 
	$O(b^2\log(n/\eps)(\log \log n + \log(a/\eps))).$
\end{theorem}

Building on \Cref{thm:CHHL}, in \cite{CHLT} Chattopadhyay, Hatami, Lovett, and Tal showed that in fact it suffices to have a bound just on $L_{1,2}({\cal F})$ in order to obtain an efficient PRG for ${\cal F}$:

\begin{theorem} [PRGs from $L_1$ Fourier norm bounds at level $k=2$: Theorem~2.1 of \cite{CHLT}] \label{thm:CHLT}
	Let ${\cal F}$ be a family of $n$-variable Boolean functions that is closed under restrictions and has $L_{1,2}({\cal F}) \leq t$.   Then there is an explicit pseudorandom generator that $\eps$-fools ${\cal F}$ with seed length 
	$O((t/\eps)^{2+o(1)}\cdot \polylog(n)).$
\end{theorem}

Observe that while \Cref{thm:CHLT} requires a weaker structural result than \Cref{thm:CHHL} (a bound only on $L_{1,2}({\cal F})$ as opposed to $L_{1,k}({\cal F})$ for all $k \geq 1$), the resulting pseudorandom generator is quantitatively weaker since it has seed length polynomial rather than logarithmic in the error parameter $1/\eps$.
Even more recently, in \cite{CGLSS} Chattopadhyay, Gaitonde, Lee, Lovett, and Shetty further developed this framework by interpolating between the two results described above. They showed that a bound on $L_{1,k}$\footnote{In fact, they showed that a bound on the weaker quantity $M_{1,k}(f) := \max_{x \in \pmone^n} \abs{\sum_{\abs{S}=k} \hf(S) x^S}$ suffices.} for any $k \ge 3$ suffices to give a PRG, with a seed length whose $\eps$-dependence scales with $k$:

\begin{theorem} [PRGs from $L_1$ Fourier norm bounds up to level $k$ for any $k$: Theorem~4.3 of \cite{CGLSS}] \label{thm:CGLSS}
  Let ${\cal F}$ be a family of $n$-variable Boolean functions that is closed under restrictions 
	and has $L_{1,k}(\calF) \le b^k$ for some $k \geq 3.$
	Then there exists a pseudorandom generator that $\eps$-fools $\calF$ with seed length 
	$
	O\left( \frac{b^{2+\frac{4}{k-2}} \cdot k \cdot \polylog(\frac{n}{\eps})}{\eps^{\frac{2}{k-2}}} \right) .
	$
\end{theorem}

\paragraph{$\F_2$-polynomials and the \cite{CHLT} conjecture.}

The works \cite{CHHL} and \cite{CHLT} highlighted the challenge of proving $L_{1,k}$ bounds for the class of \emph{bounded-degree $\F_2$-polynomials} as being of special interest.  Let 
\[
\Poly_{n,d} := \text{~the class of all $n$-variate $\F_2$-polynomials of degree $d$.}
\]
It follows from \Cref{thm:CHLT} that even proving
\begin{equation} \label{eq:weak-conj}
	L_{1,2}(\Poly_{n,\polylog(n)}) \le n^{0.49}
\end{equation}
would give nontrivial PRGs for $\F_2$-polynomials of $\polylog(n)$ degree, improving on \cite{BoV-gen,Lov09,Viola-d}.
By the classic connection (due to Razborov \cite{Raz87}) between such polynomials and the class $\AC^0[\oplus]$ of constant-depth circuits with parity gates, this would also give nontrivial PRGs, of seed length $n^{1-c}$, for $\AC^0[\oplus]$. This would be a breakthrough improvement on existing results, which are poor either in terms of seed length \cite{FeffermanSUV10}
or in terms of explicitness \cite{ChenLW2020-almost}.

The authors of \cite{CHLT} in fact conjectured the following bound, which is much stronger than  \Cref{eq:weak-conj}:
\begin{conjecture}[\cite{CHLT}] \label{conj:CHLT}
	For all $d \geq 1$, it holds that $L_{1,2}(\Poly_{n,d})=O(d^2)$.
\end{conjecture}

\paragraph{Extending the \cite{CHLT} conjecture.}
Given \Cref{conj:CHLT}, and in light of \Cref{thm:CGLSS}, it is natural to speculate that an even stronger result than \Cref{conj:CHLT} might hold.  We  consider the following natural generalization of the \cite{CHLT} conjecture, extending it from $L_{1,2}(\Poly_{n,d})$ to $L_{1,k}(\Poly_{n,d})$:

\begin{conjecture} \label{conj:CHLT-general}
	For all $d,k \geq 1$, it holds that  $L_{1,k}(\Poly_{n,d}) = O(d)^k$.
\end{conjecture}

A weaker bound of $2^{O(dk)}$ was conjectured in~\cite{CHHL}.
We note that if $p$ is the $\AND$ function on $d$ bits, then an easy computation shows that $L_{1,k}(p) = \Pr[p=1] \cdot {d \choose k}$.  Moreover, in \Cref{sec:reduction} we provide a reduction showing that this implies that the upper bound $O(d)^k$ in \Cref{conj:CHLT-general} is best possible for any constant $k$.

For positive results, the work \cite{CHLT} proved that $L_{1,1}(\Poly_{n,d}) \leq 4d$ and $L_{1,2}(\Poly_{n,d}) \le O(d \sqrt{n\log n})$, and already in \cite{CHHL} it was shown that $L_{1,k}(\Poly_{n,d}) \leq (2^{3d}\cdot k)^k,$ but to the best of our knowledge no other results towards \Cref{conj:CHLT} or \Cref{conj:CHLT-general} are known.  

Given the apparent difficulty of resolving \Cref{conj:CHLT} and \Cref{conj:CHLT-general} in the general forms stated above, it is natural to study $L_{1,2}$ and $L_{1,k}$ bounds for specific subclasses of degree-$d$ $\F_2$-polynomials.  This study is the subject of our main structural results, which we describe in the next subsection.

\subsection{Our results: Fourier bounds for structured $\F_2$-polynomials}

Our main results show that $L_{1,2}$ and $L_{1,k}$ bounds of the flavor of \Cref{conj:CHLT} and \Cref{conj:CHLT-general} indeed hold for several well-studied  classes of $\F_2$-polynomials, specifically \emph{symmetric} $\F_2$-polynomials and \emph{read-$\Delta$} $\F_2$-polynomials.  We additionally prove a composition theorem that allows us to combine such polynomials (or, more generally, any polynomials that satisfy certain $L_{1,k}$ bounds) in a natural way and obtain $L_{1,k}$ bounds on the resulting combined polynomials.

Before describing our results in detail, we pause to briefly explain why (beyond the fact that they are natural mathematical objects) such ``highly structured'' polynomials are attractive targets of study  given known results.  It has been known for more than ten years~\cite[Lemma~2]{BHL12} that for any degree $d<(1-\eps)n$, a \emph{random} $\F_{2}$-polynomial of degree $d$ (constructed by independently including each monomial of degree at most $d$ with probability $1/2$) is extremely unlikely to have bias larger than $\exp(-n/d)$.
 It follows that as long as $d$ is not too large, a random degree-$d$
polynomial $p$ is overwhelmingly likely to have $L_{1,k}(p)=o_n(1)$, which is much smaller than $d^{k}$.
(To verify this, consider the polynomials $p_S$ obtained by XORing $p$ with the parity function $\sum_{i\in S} x_i$.
Note that the bias of $p_S$ is the Fourier coefficient of $(-1)^p$ on $S$.
Now apply \cite[Lemma~2]{BHL12} to each polynomial $p_S$, and sum the terms.)

Since the conjectures hold true for random polynomials, it is natural
to investigate highly structured polynomials.

\subsubsection{Symmetric $\F_2$-polynomials}

A \emph{symmetric} $\F_2$-polynomial over $x_1,\dots,x_n$ is one whose output depends only on the Hamming weight of its input $x$.  Such a polynomial of degree $d$ can be written in the form
\[
p(x) := \sum_{k=0}^d c_k \sum_{\abs{S}=k, S \subseteq [n]} \prod_{i\in S} x_i ,
\]
where $c_0, \ldots, c_d \in \zo$.
While symmetric polynomials may seem like simple objects, their study can sometimes lead to unexpected discoveries; for example, a symmetric, low-degree $\F_2$-polynomial provided a counterexample
to the ``Inverse conjecture for the Gowers norm'' \cite{LovettMS11,GT}.
For another example \cite{corr}, symmetric polynomials have been conjectured to have maximal correlation with the mod 3 function.  The conjecture has been verified for degree 2, and if true for higher degrees it would imply long-sought progress in circuit complexity.

We prove the following upper and lower bounds on the $L_1$ Fourier norm at level $k$ for any symmetric polynomial:
\begin{theorem} \label{thm:symm}
	Let $p(x_1,\dots,x_n)$ be a symmetric $\F_2$-polynomial of degree $d$.
	Then $L_{1,k}(p) \le \Pr[p = 1] \cdot O(d)^k$ for every $k$.
\end{theorem}

\begin{theorem} \label{thm:sym-lb}
  For every $k$, there is a symmetric $\F_2$-polynomial $p(x_1, \ldots, x_n)$ of degree $d = \Theta(\sqrt{kn})$ such that $L_{1,k}(p) \ge (e^{-k}/2) \cdot \binom{n}{k}^{1/2} = \Omega_k(1) \cdot d^k$.
\end{theorem}
\Cref{thm:symm} verifies the \cite{CHLT} conjecture (\Cref{conj:CHLT}), and even the generalized version \Cref{conj:CHLT-general}, for the class of symmetric polynomials.  \Cref{thm:sym-lb} complements it by showing that the upper bounds in \Cref{thm:symm} are tight for $k = O(1)$ when $d = \Theta(\sqrt{n})$.

\ignore{
}
\Cref{thm:symm} also provides a quadratic sharpening of an earlier bound that was implicit in \cite{RSV13} (as well as providing the ``correct'' dependence on $\Pr[p=1]$).
In \cite{RSV13} Reingold, Steinke and Vadhan showed that any function $f$ computed by an oblivious, read-once, regular branching program of width $w$ has $L_{1,k}(f) \leq (2w^2)^k$. It follows directly from a result of \cite{BGL} (\Cref{lemma:BGL} below) that any symmetric $\F_2$-polynomial $p$ of degree $d$ can be computed by an oblivious, read-once, regular branching program of width at most $2d$, and hence the \cite{RSV13} result implies that $L_{1,k}(p) \leq 8^k d^{2k}.$

Subsequent to our work, Lee, Pyne, and Vadhan~\cite{LPV22} improved the $(2w^2)^k$ bound obtained in~\cite{RSV13} to $\Pr[f=1] (w-1)^k$, which recovers \Cref{thm:symm}.

\subsubsection{Read-$\Delta$ $\F_2$-polynomials}

For $\Delta \geq 1$, a \emph{read-$\Delta$} $\F_2$-polynomial is one in which each input variable appears in at most $\Delta$ monomials.
The case $\Delta=1$ corresponds to the class of \emph{read-once polynomials}, which are simply sums of monomials over disjoint sets of variables; for example, the polynomial $x_1x_2 + x_3x_4$ is read-once whereas $x_1x_2 + x_1x_4$ is read-twice.  Read-once polynomials have been studied from the perspective of pseudorandomness~\cite{LV,MRT,Lee,DHH} as they capture several difficulties in improving Nisan's generators~\cite{Nis92} for width-$4$ read-once branching programs.

We show that the $L_{1,k}$ Fourier norm of read-$\Delta$ polynomials is polynomial in $d$ and $\Delta$:
\begin{theorem} \label{thm:read-few-poly}
  Let $p(x_1, \ldots, x_n)$ be a read-$\Delta$ polynomial of degree $d$.
  Then $L_{1,k}(p) \le \Pr[p=1] \cdot O(k)^k \cdot (d\Delta)^{10k}$.
\end{theorem}

The work \cite{Lee} showed that read-once polynomials satisfy an $L_{1,k}$ bound of $O(d)^k$ and this is tight for every $k$, but we are not aware of previous bounds on even the $L_1$ Fourier norm at level $k=2$ for read-$\Delta$ polynomials, even for $\Delta=2$.

As any monomial with degree $\Omega(\log n)$ vanishes under a random restriction with high probability, we have the following corollary which applies to polynomials of any degree.

\begin{corollary} \label{cor:L1k-read-few-poly}
  Let $p(x_1, \ldots, x_n)$ be a read-$\Delta$ polynomial.
  Then $L_{1,k}(p) \le O(k)^{11k} \cdot (\Delta \log n)^{10k}$.
\end{corollary}






\subsubsection{A composition theorem}

The upper bounds of \Cref{thm:symm} and \Cref{thm:read-few-poly} both include a factor of $\Pr[p=1]$.  (We observe that negating $p$, i.e.~adding 1 to it, does not change its $L_{1,2}$ or $L_{1,k}$ and keeps $p$ symmetric (respectively, read-$\Delta$) if it was originally symmetric (respectively, read-$\Delta$), and hence in the context of those theorems we can assume that this $\Pr[p=1]$ factor is at most $1/2.$) 
Level-$k$ bounds that include this factor have appeared in earlier works for other classes of functions~\cite{OS,BTW,CHRT,Tal,GirishTW21}, and have been used to obtain high-level bounds for other classes of functions~\cite{CHRT,Tal,GirishTW21} and to extend level-$k$ bounds to more general classes of functions~\cite{Lee}.\ignore{
}
Having these $\Pr[p=1]$ factors in \Cref{thm:symm} and \Cref{thm:read-few-poly} is important for us in the context of our composition theorem, which we now describe.  We begin by defining the notion of a \emph{disjoint composition} of functions:

\begin{definition} \label{def:disjoint-composition}
	Let $\calF$ be a class of functions from $\pmo^m$ to $\pmo$ and let $\calG$ be a class of functions from $\pmo^\ell$ to $\pmo$. 
	Define the class $\calH = \calF \circ \calG$ of \emph{disjoint compositions of $\calF$ and $\calG$} to be the class of all functions from $\pmo^{m \ell}$ to $\pmo$ of the form
	\[
	h(x^1, \ldots, x^m) = f(g_1(x^1), \ldots, g_m(x^m)),
	\]
	where $g_1, \ldots, g_m \in \calG$ are defined on $m$ disjoint sets of variables and $f \in \calF$.
\end{definition}

As an example of this definition, the class of \emph{block-symmetric} polynomials (i.e.~polynomials whose variables are divided into blocks and are symmetric within each block but not overall) are a special case of disjoint compositions where ${\cal G}$ is taken to be the class of symmetric polynomials. We remark that block-symmetric polynomials are known to correlate better with parities than symmetric polynomials in certain settings~\cite{GKV}.

We prove a composition theorem for upper-bounding the $L_1$ Fourier norm at level $k$ of the disjoint composition of any classes of functions that are closed under restriction and admit a $L_{1,k}$ bound of the form $\Pr[f=1] \cdot a \cdot b^k$:

\begin{theorem} \label{thm:composition}
	Let $g_1, \ldots, g_m \in \calG$ and let $f \in \calF$, where $\calF$ is closed under restrictions.
	Suppose that for every $1 \le k \le K$, we have
	\begin{enumerate}
		\item $L_{1,k}(f) \le \Pr[f=1] \cdot a_\outer \cdot b_\outer^k$ for every $f \in \calF$, and
		\item $L_{1,k}(g) \le \Pr[g=1] \cdot a_\inner \cdot b_\inner^k$ for every $g \in \calG$.
	\end{enumerate}
	Then for every $\pm 1$-valued function $h \in \calH = \calF \circ \calG$, we have that   \[
	L_{1,K}(h)
  \le  \Pr[h=1] \cdot a_\outer \cdot (a_\inner b_\inner b_\outer)^K .
	\]
\end{theorem}

See \Cref{thm:L1-of-composition} for a slightly sharper bound.
We remark that \Cref{thm:composition} does not assume any $\F_2$-polynomial structure for the functions in ${\cal F}$ or ${\cal G}$ and thus may be of broader utility.

\subsection{Applications of our results}
	Our structural results imply new pseudorandom generators and correlation bounds.
	\subsubsection{Pseudorandom generators}
	
	Combining our Fourier bounds with the polarizing framework, we obtain new PRGs for read-few polynomials.
	The following theorem follows from applying \Cref{thm:CGLSS} with some $k = \Theta(\log n)$ and the $L_{1,k}$ bound in \Cref{cor:L1k-read-few-poly}.

  \begin{theorem}
    There is an explicit pseudorandom generator that $\eps$-fools read-$\Delta$ $\F_2$-polynomials with seed length $\poly(\Delta, \log n, \log(1/\eps))$.
  \end{theorem}
  For constant $\eps$, this improves on a PRG by Servedio and Tan~\cite{ServedioT-PRG-sparse}, which has a seed length of $2^{O(\sqrt{\log(\Delta n)})}$.
  (Note that read-$\Delta$ polynomials are also $(\Delta n)$-sparse.)
  We are not aware of any previous PRG for read-2 polynomials with $\polylog(n)$ seed length.

  Note that the $\OR$ function has $L_1$ Fourier norm $O(1)$.
  By expressing a DNF in the Fourier expansion of $\OR$ in its terms, it is not hard to see that the same PRG also fools the class of read-$\Delta$ DNFs (and read-$\Delta$ CNFs similarly)~\cite{ServedioT-PRG-readfew}.
  	
	
\subsubsection{Correlation bounds}

Exhibiting explicit Boolean functions that do not \emph{correlate} with low-degree polynomials is a fundamental challenge in complexity.  Perhaps surprisingly, this challenge stands in the way of progress on a striking variety of frontiers in complexity, including circuits, rigidity, and multiparty communication complexity.  For a survey of correlation bounds and discussions of these connections we refer the reader to
\cite{viola-FTTCS09,Viola-map,Viola-diago}.

For polynomials of degree larger than $\log_2 n$, the state-of-the-art remains the lower bound proved by Razborov and Smolensky in the 1980s' \cite{Raz87,Smolensky87}, showing that for any degree-$d$ polynomial $p$ and an explicit function $h$ (in fact, majority) we have:
\[
\Pr[ p(x) = h(x)] \le 1/2 + O(d/\sqrt{n}).\]

Viola~\cite{Viola20} recently showed that upper bounds on $L_{1,k}(\calF)$ imply correlation bounds between $\calF$ and an explicit function $h_k$ that is related to majority and is defined as

\[h_k(x) := \sgn\Biggl(\sum_{\abs{S}=k} x^S\Biggr).\]

In particular, proving \Cref{conj:CHLT} or related conjectures implies new correlation bounds beating Razborov--Smolensky.  The formal statement of the connection is given by the following theorem.

\begin{theorem}[Theorem~1 in \cite{Viola20}] \label{viola-L12-implies-cor}
  For every $k \in [n]$ and $\calF \subseteq \{f\colon \zo^n \to \pmone\}$, there is a distribution $D_k$ on $\zo^n$ such that for any $f \in \calF$,
	\[
	\Pr_{x \sim D_k} \left[f(x) = h_k(x)\right]
  \le \frac{1}{2} + \frac{e^k}{2 \sqrt{\binom{n}{k}}} L_{1,k}(\calF) .
	\]
\end{theorem}

For example, if $k=2$ and we assume that the answer to \Cref{conj:CHLT} is positive, then the right-hand side above becomes $1/2 + O(d^2/n)$, which is a quadratic improvement over the bound by Razborov and Smolensky.


Therefore, \Cref{thm:symm,thm:read-few-poly} imply correlation bounds between these polynomials and an explicit function that are better than $O(d/\sqrt{n})$ given in \cite{Raz87,Smolensky87}.
We note that via a connection in \cite{Viola-d}, existing PRGs for these polynomials already imply strong correlation bounds between these polynomials and the class of $\NP$.
Our results apply to more general classes via the composition theorem, where it is not clear if previous techniques applied.  For a concrete example, consider the composition of a degree-$(n^\alpha)$ symmetric polynomial with degree-$(n^\alpha)$ read-$(n^\alpha)$ polynomials.  \Cref{thm:composition} shows that such polynomial has $L_{1,2} \le n^{O(\alpha)}$.  For a sufficiently small $\alpha = \Omega(1)$, we again obtain correlation bounds improving on Razborov--Smolensky.

%
%

\subsection{Related work}

We close this introduction by discussing a recent work of Girish, Tal and Wu~\cite{GirishTW21} on parity decision trees that is related to our results.

Parity decision trees are a generalization of decision trees in which each node queries a parity of some input bits rather than a single input bit.
The class of depth-$d$ parity decision trees is a subclass of $\F_2$ degree-$d$ polynomials, as such a parity decision tree can be expressed as a sum of products of sums over $\F_2$, where each product corresponds to a path in the tree (and hence gives rise to $\F_2$-monomials of degree at most $d$).
The Fourier spectrum of parity decision trees was first studied in \cite{BTW}, which obtained
a level-$1$ bound of $O(\sqrt{d})$.
This bound was recently extended to higher levels in \cite{GirishTW21}, showing that any depth-$d$ parity decision tree $T$ over $n$ variables has $L_{1,k}(T) \leq d^{k/2} \cdot O(k \log n)^k$.



\section{Our techniques}

We now briefly explain the approaches used to prove our results.
We note that each of these results is obtained using very different ingredients, and hence the results can be read independently of each other.

\subsection{Symmetric polynomials (\Cref{thm:symm,thm:sym-lb}, \Cref{sec:symmetric})}
The starting point of our proof is a result from~\cite{BGL}, which says that degree-$d$ symmetric $\F_2$-polynomials only depend on the Hamming weight of their input modulo $m$ for some $m$ (a power of two) which is  $\Theta(d)$, and the converse is also true.  (See \Cref{lemma:BGL} for the exact statement.)
We now explain how to prove our upper and lower bounds given this.

For the upper bound (\Cref{thm:symm}), since $p(x)$ takes the same value for all strings $x$ with the same weights $\ell \bmod m$, to analyze $L_{1,k}(p)$ it suffices to analyze $\E[(-1)^{x_1 + \cdots + x_k}]$ conditioned on $x$ having Hamming weight exactly $\ell \bmod m$.

We bound this conditional expectation by considering separately two cases depending on whether or not $k \le n/m^2$.
For the case that $k \le n/m^2$, we use a (slight sharpening of a) result from \cite{BHLV}, which gives a bound of $m^{-k} e^{-\Omega(n/m^2)}$.
In the other case, that $k \ge n/m^2$, in \Cref{lemma:bias-large-m} we prove a bound of $O(km/n)^k$. This is established via a careful argument that gives a new bound on the Kravchuk polynomial in certain ranges (see \Cref{claim:kravchuk-bound}), extending and sharpening similar bounds that were recently established in \cite{CohenPT21} (the bounds of \cite{CohenPT21} would not suffice for our purposes). 

In each of the above two cases, summing over all the $\binom{n}{k}$ coefficients gives the desired bound of $O(m)^k = O(d)^k$.

For the lower bound (\Cref{thm:sym-lb}), we first consider the symmetric function $h_k(x) := \sgn(\sum_{\abs{S}=k} (-1)^{\sum_{i\in S} x_i})$.  (Note that $h_1$ is the majority function.)
As $h_k$ is the sign of a degree-$k$ real polynomial (i.e. a degree-$k$ polynomial threshold function), it follows from a consequence of the hypercontractive inequality that $L_{1,k}(h_k) \ge e^{-k} \binom{n}{k}^{1/2}$.
We define our polynomial $p$ so that on inputs $x$ whose Hamming weight $\abs{x}$ is within $m=\Theta(\sqrt{kn})$ from $n/2$, it agrees with $h_k$, and on other inputs its value is defined such that $p$ only depends on $\abs{x} \bmod 2m$.  By the converse of \cite{BGL}, $p$ can be computed by a symmetric polynomial of degree less than $2m$.
Our lower bound then follows from $L_{1,k}(f) \le \binom{n}{k}^{1/2}$ for any Boolean function $f$, and that by our choice of $m$, the two functions $p$ and $h_k$ disagree on $e^{-\Omega(k)}$ fraction of the inputs.

\subsection{Read-$\Delta$ polynomials (\Cref{thm:read-few-poly}, \Cref{sec:read-few})}
Writing $f := (-1)^p$ for an $\F_2$-polynomial $p$, we observe that the coefficient $\hf(S)$ is simply the bias of $p_S(x) := p(x) + \sum_{i \in S} x_i$.
Our high-level approach is to decompose the read-few polynomial $p_S$ into many disjoint components, then show that each component has small bias.
Since the components are disjoint, the product of these biases gives an upper bound on the bias of $p_S$.

In more detail, we first partition the variables according to the minimum degree $t_i$ of the monomials containing each variable $x_i$.
Then we start decomposing $p_S$ by collecting all the monomials in $p$ containing $x_i$ to form the polynomial $p_i$.
We observe that the larger $t_i$ is, the more likely $p_i$ is to vanish on a random input, and therefore the closer $p_i + x_i$ is to being unbiased.
For most $S$, we can pick many such $p_i$'s ($i\in S$) from $p$ so that they are disjoint.
For the remaining polynomial $r$, because $\Delta$ and $d$ are small, we can further decompose $r$ into many disjoint polynomials $r_i$.
Finally, our upper bound on $|\hf(S)|$ will be the magnitude of the product of the biases of the $p_i$'s and $r_i$'s.
We note that our decomposition of $p$ uses the structure of $S$; and so the upper bound on $\hf(S)$ depends on $S$ (see \Cref{lemma:read-few-key}).
Summing over each $\abs{\hf(S)}$ gives our upper bound.

\subsection{Composition theorem (\Cref{thm:composition}, \Cref{sec:compositions})}
As a warmup, let us first consider directly computing a degree-1 Fourier coefficient $\hh(\{(i,j)\})$ of the composition.
Since the inner functions $g_i$ depend on disjoint variables, by writing the outer function $f$ in its Fourier expansion, it is not hard to see that 
\[
  \hh(\{(i,j)\})
  = \sum_{S \ni i} \hf(S) \prod_{\ell \in S \setminus\{i\}} \E[g_\ell] \cdot \hg_i(\{j\}) .
\]
When the $g_i$'s are balanced, i.e.\ $\E[g_i] = 0$, we have $\hf(\{(i,j)\}) = \hf(\{i\}) \hg_i(\{j\})$, and it follows that $L_{1,1}(h) \le L_{1,1}(\calF) L_{1,1}(\calG)$.
To handle the unbalanced case, we apply an idea from \cite{CHHL} that lets us relate $\sum_{S \ni i} \hf(S) \prod_{\ell \in S \setminus\{i\}} \E[g_\ell]$ to the average of $\widehat{f_R}(\{i\})$, for some suitably chosen random restriction $R$ on $f$ (see \Cref{claim:derivatives}).
As $\calF$ is closed under restrictions, we can apply the $L_{1,1}(\calF)$ bound on $f_R$, which in turns gives a bound on $\sum_{S \ni i} \hf(S) \prod_{\ell \in S \setminus\{i\}} \E[g_\ell]$ in terms of $L_{1,1}(\calF)$ and $\E[g_i]$.

Bounding $L_{1,k}(h)$ for $k \ge 2$ is more complicated, as each $\hh(S)$ involves $\hf(J)$ and $\hg_i(T)$'s, where the sets $J$ and $T$ have different sizes.
We provide more details in \Cref{sec:compositions}.




\section{Preliminaries}

\noindent {\bf Notation.}
For a string $x \in \zo^n$ we write $\abs{x}$ to denote its Hamming weight $\sum_{i=1}^n x_i$.
We use $\calX_w$ to denote $\{x : \abs{x} = w\}$, the set of $n$-bit strings with Hamming weight $w$, and $\calX_{\ell \bmod m} = \bigcup_{w: w \equiv \ell \bmod m} \calX_w = \{x : \abs{x} \equiv \ell \bmod m\}$.

We recall that for an $n$-variable Boolean function $f$, the \emph{level-$k$ Fourier $L_1$ norm of $f$} is
\[
  L_{1,k}(f) = \sum_{S \subset [n]: \abs{S} = k} \abs{\hf(S)}.
\]
We note that a function $f$ and its negation have the same $L_{1,k}$ for $k \ge 1$.
Hence we can often assume that $\Pr[f=1] \le 1/2$, or replace the occurrence of $\Pr[f=1]$ in a bound by $\min\{\Pr[f=1],\Pr[f=0]\}$ for a $\zo$-valued function $f$ (or by $\min\{\Pr[f=1],\Pr[f=-1]\}$ for a $\pmo$-valued function). 
If $f$ is a $\pmo$-valued function then $\frac{1 - \abs{\E[f]}}{2}$ is equal to $\min\{\Pr[f=1],\Pr[f=-1]\}$, and we will often write $\frac{1 - \abs{\E[f]}}{2}$ for convenience.

Unless otherwise indicated, we will use the letters $p,q,r$, etc.~to denote $\F_2$-polynomials (with inputs in $\zo^n$ and outputs in $\zo$) and the letters $f,g,h,$ etc.~to denote general Boolean functions (where the inputs may be $\zo^n$ or $\pmo^n$ and the outputs may be $\zo$ or $\pmo$ depending on convenience).
We note that changing from $\zo$ outputs to $\pmo$ outputs only changes $L_{1,k}$ by a factor of $2$. 

We use standard multilinear monomial notation as follows: given a vector $\beta = (\beta_1,\dots,\beta_n)$ and a subset $T \subseteq [n]$, we write $\beta^T$ to denote $\prod_{j \in T} \beta_j$.


\section{\texorpdfstring{$L_{1,k}$ bounds for symmetric polynomials}{}} \label{sec:symmetric}

The main result of this section is \Cref{thm:L1-sym}, which gives an upper bound on $L_{1,k}(p)$ for any symmetric $\F_2$-polynomial $p$ of degree $d$, covering the entire range of parameters $1 \leq k,d \leq n$:

\begin{theorem}[Restatement of \Cref{thm:symm}] \label{thm:L1-sym}
	Let $p\colon \zo^n \to \zo$ be a symmetric $\F_2$-polynomial of degree $d$.
	For every $1 \leq k,d \leq n$, 
	\[
	  L_{1,k}(p) := \sum_{\abs{S}=k} \abs{\hp(S)} \le \Pr[p(x) = 1] \cdot O(d)^k .
  \]
\end{theorem}


%

\paragraph{Proof idea.}
As the polynomial $p$ is symmetric, its Fourier coefficient $\hp(S)$ only depends on $|S|$, the size of $S$.
Hence to bound $L_{1,k}$ it suffices to analyze the coefficient $\hp(\{1, \ldots, k\}) = \E_{x \sim \zo^n}[p(x) (-1)^{x_1 + \cdots + x_k}]$.

Our proof uses a  result from~\cite{BGL} (\Cref{lemma:BGL}  below), which says that degree-$d$ symmetric $\F_2$-polynomials only depend on the  Hamming weight of their input modulo $m$ for some $m = O(d)$.
Given this, since $p(x)$ takes the same value for strings $x$ with the same weights $\ell \bmod m$, we can in turn bound each $\E[(-1)^{x_1 + \cdots + x_k}]$ conditioned on $x$ having Hamming weight exactly $\ell \bmod m$, i.e.\ $x \in \calX_{\ell \bmod m}$.
We consider two cases depending on whether or not $k \le n/m^2$.
If $k \le n/m^2$, we can apply a (slight sharpening of a) result from \cite{BHLV}, which gives a bound of $m^{-k} e^{-\Omega(n/m^2)}$.
If $k \ge n/m^2$, in \Cref{lemma:bias-large-m} we prove a bound of $O(km/n)^k$.
In each case, summing over all the $\binom{n}{k}$ coefficients gives the desired bound of $O(m)^k = O(d)^k$.

We now give some intuition for \Cref{lemma:bias-large-m}, which upper bounds the magnitude of the ratio
\begin{align} \label{eq:ratio}
  \E_{x \sim \calX_{\ell \bmod m}} [(-1)^{x_1 + \cdots + x_k}]
  = \frac{\sum_{x \in \calX_{\ell \bmod m}} (-1)^{x_1 + \cdots + x_k}}{\abs{\calX_{\ell \bmod m}}}
\end{align}
by $O(km/n)^k$.
Let us first consider $k = 1$ and $m = \Theta(\sqrt{n})$.
As most strings $x$ have Hamming weight within $[n/2-\Theta(\sqrt{n}),n/2+\Theta(\sqrt{n})]$, it is natural to think about the weight $\abs{x}$ in the form of $n/2 + m\Z + \ell'$.
It is easy to see that the denominator is at least $\Omega(2^n/\sqrt{n})$,
so we focus on bounding the numerator.
Consider the quantity $\sum_{x \in \calX_{n/2 + s}} \E[(-1)^{x_1}]$ for some $s$.
As we are summing over all strings of the same Hamming weight, we can instead consider $\sum_{x \in \calX_{n/2+s}} \E_{i \sim [n]}[(-1)^{x_i}]$ (see \Cref{eq:kravchuk-identity}).
For any string of weight $n/2 + s$, it is easy to see that 
\begin{align} \label{eq:kravchuk-1}
  \E_{i \sim [n]}[(-1)^{x_i}] = (1/2 - s/n) - (1/2 + s/n) = -2s/n .
\end{align}
Therefore, in the $k=1$ case we get that
\[
  \abs[\bigg]{\E_{x \sim \calX_{\ell \bmod m}} [(-1)^{x_1 + \cdots + x_k}]}
  \le 2 \sum_c \binom{n}{n/2 + cm + \ell'} \frac{\abs{cm+\ell'}}{n} .
\]
Using the fact that $\binom{n}{n/2 + cm + \ell'}$ is exponentially decreasing in $\abs{c}$, in \Cref{claim:moment-mod-m} we show that this is at most $O(2^n/n)$.
So the ratio in \cref{eq:ratio} is at most $O(1/\sqrt{n})$, as desired, when $k=1.$

However, already for $k=2$, a direct (but tedious) calculation shows that
\begin{align} \label{eq:kravchuk-2}
  \E_{i < j}[(-1)^{x_i + x_j}]
  = \frac{4s^2 - 2ns + n}{n(n-1)} ,
\end{align}
which no longer decreases in $s$ like in \cref{eq:kravchuk-1}.
Nevertheless, we observe that this is bounded by $O(1/n + (\abs{s}/n)^2)$, which is sufficient for bounding the ratio by $O(1/n)$.
Building on this, for any $k$ we obtain a bound of $2^{O(k)} ((k/n)^{k/2} + (\abs{s}/n))^k$ in \Cref{claim:kravchuk-bound}, and by a more careful calculation we are able to obtain the desired bound of $O(km/n)^k$ on \Cref{eq:ratio}.

\subsection{Proof of \Cref{thm:L1-sym}}
We now prove the theorem.
We will use the following result from~\cite{BGL}, which says that degree-$d$ symmetric 
$\F_2$-polynomials only depend on their input's Hamming weight modulo 
$O(d)$.
\begin{lemma}[{{Corollaries~2.5 and 2.7 in~\cite{BGL}, $p=2$}}] \label{lemma:BGL}
	Let $p\colon\zo^n \to \zo$ be any function and $m$ be a power of two.
	If $p$ is a symmetric $\F_2$-polynomial of degree $d$, where $m/2 \le d < m$, then $p(x)$ only depends on $\abs{x} \bmod m$.
  Conversely, if $p(x)$ only depends on $\abs{x} \bmod m$, then it can be computed by a symmetric $\F_2$-polynomial of degree less than $m$.
\end{lemma}
We will also use two bounds on the biases of parities under the uniform distribution over $\calX_{\ell \bmod m}$, one holds for $k \le n/(2d)^2 \le n/m^2$ (\Cref{claim:bias-small-m}) and the other for $k \ge n/(2d)^2 \ge n/(4m^2)$ (\Cref{lemma:bias-large-m}).
\Cref{claim:bias-small-m} is essentially taken from \cite{BHLV}.
However, the statement in \cite{BHLV} has a slightly worse bound, so for completeness we provide a self-contained proof below to give the bound of \Cref{claim:bias-small-m}.
The proof of \Cref{lemma:bias-large-m} involves bounding the magnitude of Kravchuk polynomials.
As it is somewhat technical we defer its proof to \Cref{sec:kravchuk}.

\begin{claim}[Lemma~10 in \cite{BHLV}] \label{claim:bias-small-m}
	For every $1 \le k \le n/m^2$ and every integer $\ell$,
	\[
	2^{-n} \abs[\Big]{\sum_{x \in \calX_{\ell \bmod m}} (-1)^{x_1 + \cdots + x_k}} 
  \le m^{-(k+1)} e^{-\Omega(n/m^2)},
	\]
	while for $k=0$,
	\[
	\abs[\Big]{ 2^{-n} \abs{\calX_{\ell \bmod m}} - 1/m }
  \le m^{-1} e^{-\Omega(n/m^2)} .
	\]
\end{claim}

\begin{lemma} \label{lemma:bias-large-m}
  For $k \ge n/(4m^2)$, we have
	\[
	\binom{n}{k} \cdot \max_\ell \abs*{ \frac{\sum_{x \in \calX_{\ell \bmod m}}  (-1)^{x_1 + \cdots + x_k}}{\abs{\calX_{\ell \bmod m}}} }
  \le O(m)^k  .
	\]
\end{lemma}

We now use \Cref{claim:bias-small-m} and \Cref{lemma:bias-large-m} to prove \Cref{thm:L1-sym}.

\begin{proof}[Proof of \Cref{thm:L1-sym}]
	As $p$ is symmetric, all the level-$k$ coefficients are the same, so
	it suffices to give a bound on $\hp(\{1,2, \ldots, k\})$.
  Let $\tilde p\colon \{0, \ldots, n\} \to \zo$ be the function defined by $\tilde 
	p(\abs{x}) := p(x_1, \ldots, x_n)$.
	By \Cref{lemma:BGL}, we have $\tilde p(\ell) = \tilde p(\ell 
	\bmod m)$ for some $d < m \le 2d$ where $m$ is a power of $2$.
	Using the definition of $\hp(\{1, \ldots, k\})$, we have
	\begin{align*}
	\abs{ \hp(\{1, \ldots, k\}) }
	&= \abs*{ \E_{x \sim \zo^n} \bigl[ p(x) (-1)^{x_1 + \cdots + x_k} 
		\bigr] } \\
	&= \abs*{ \sum_{\ell=0}^{m-1} \tilde p(\ell) \frac{\abs{\calX_{\ell \bmod m}}}{2^n} \cdot \frac{\sum_{x \in \calX_{\ell \bmod m}}  (-1)^{x_1 + \cdots + x_k}}{\abs{\calX_{\ell \bmod m}}} } \\
	&\le \E[p] \cdot \max_{0 \le \ell \le m-1} \abs*{ \frac{\sum_{x \in \calX_{\ell \bmod m}}  (-1)^{x_1 + \cdots + x_k}}{\abs{\calX_{\ell \bmod m}}} } ,
	\end{align*}
	where we use the shorthand $\E[p] = \E_{x \sim \zo^n} [p(x)]$ in the last step.
	
  When $k \le n/(2d)^2 \le n/m^2$, by \Cref{claim:bias-small-m} (using the first bound for the numerator and the second $k=0$ bound for the denominator) we have 
	\[
	\max_{0 \le \ell \le m-1} \abs*{ \frac{\sum_{x \in \calX_{\ell \bmod m}}  (-1)^{x_1 + \cdots + x_k}}{\abs{\calX_{\ell \bmod m}}} }
	\le \frac{m^{-(k+1)} e^{-\Omega(n/m^2)}}{m^{-1} (1 - e^{-\Omega(n/m^2)})} \\
	\le O(1) \cdot m^{-k} e^{-\Omega(n/m^2)} ,
	\]
	where the last inequality holds because $1 \leq k \leq n/m^2$ and hence the $(1 - e^{-\Omega(n/m^2)})$ factor in the denominator of the left-hand side is $\Omega(1)$.
	Hence, summing over all the $\binom{n}{k}$ level-$k$ coefficients, we get that
	\[
	L_{1,k}(p)
	\le \E[p] \cdot \binom{n}{k} \cdot O(1) \cdot m^{-k} e^{-\Omega(n/m^2)}
	\le \E[p] \cdot O(1) \cdot m^k \left(\frac{ne}{k m^2}\right)^k e^{-\Omega(n/m^2)} 
  \le \E[p] \cdot O(m)^k ,
	\]
  where the last inequality is because for constant $c$, the function $(x/k)^k e^{-cx}$ is maximized when $x = k/c$, and is $O(1)^k$.

  When $k \ge n/(2d)^2 \ge n/(4m^2)$, by \Cref{lemma:bias-large-m} we have
	\[
	L_{1,k}(p)
	\le \E[p] \cdot \binom{n}{k} \max_{0 \le \ell \le m-1} \abs*{ \frac{\sum_{x \in \calX_{\ell \bmod m}}  (-1)^{x_1 + \cdots + x_k}}{\abs{\calX_{\ell \bmod m}}} }
  \le \E[p] \cdot O(m)^k . \qedhere
	\]
\end{proof} 

It remains to prove \Cref{claim:bias-small-m} and \Cref{lemma:bias-large-m}.
\begin{proof}[Proof of \Cref{claim:bias-small-m}]
  Consider an expansion of
  \[
    p(y) = (1-y)^k (1+y)^{n-k}
  \]
  into $2^n$ terms indexed by $x \in \zo^n$ where $x_i = 0$ indicates that we take the term 1 from the $i$-th factor.
  It is easy to see that the coefficient of $y^d$ is
  $\sum_{\abs{x}=d} (-1)^{\sum_{i=1}^k x_i}$, where $\abs{x}$ denotes the number of occurrences of $1$ in $x$.
  Denote $\zeta := e^{2\pi i/m}$ as an $m$-th root of unity.
  Recall the identity
  \[
    \frac{1}{m} \sum_{j=0}^{m-1} \zeta^{j (d-\ell)}
    = \begin{cases} 1 & \text{if $d \equiv \ell \pmod m$} \\
                    0 & \text{otherwise.} \end{cases}
  \]
  Thus the sum we want to bound is equal to
  \[
    \frac{1}{m}\sum_{j=0}^{m-1} \zeta^{-j \ell} p(\zeta^j) .
  \]
  Note that for $j=0$ we have $p(\zeta^0) = p(1) = 0$ for $k \neq 0$ while $p(\zeta^0)=2^n$ for $k = 0$.
  We now bound above $p(\zeta^j)$ for $j > 0$.
  As $\abs{1 + e^{i\theta}} = 2\abs{\cos (\theta/2)}$ and $\abs{1 - e^{i\theta}} = 2\abs{\sin (\theta/2)}$ we have
  \[
    \abs{p(\zeta^j)}
    = \abs{1-\zeta^j}^k \abs{1+\zeta^j}^{n-k}
    = 2^{n} \Bigl| \sin\frac{j \pi}{m} \Bigr|^k \Bigl| \cos\frac{j \pi}{m} \Bigr|^{n-k} .
  \]
Therefore, for $k \ne 0$ we have
\[
	\abs*{ \sum_{x \in \calX_{\ell \bmod m}} (-1)^{x_1 + \cdots + x_k} }
  \le \frac{1}{m} \sum_{j=0}^{m-1} \abs{ \zeta^{j\ell} \cdot p(\zeta^j) }
  \le \frac{2^n}{m} \sum_{j=0}^{m-1} \abs[\Big]{ \cos\frac{j\pi}{m} }^{n-k} \abs[\Big]{\sin\frac{j\pi}{m}}^k .
\]
For $0 < \theta < \pi/2$, we have $\cos\theta \le e^{-\theta^2/2}$ 
and $\sin\theta \le \theta$.
So
\begin{align*}
\sum_{j=0}^{m-1} \abs[\Big]{ \cos\frac{j\pi}{m} }^{n-k} \abs[\Big]{\sin\frac{j\pi}{m}}^k
&= 2 \sum_{0 < j < m/2} \Bigl( \cos\frac{j\pi}{m} \Bigr)^{n-k} \Bigl(\sin\frac{j\pi}{m}\Bigr)^k \\
&\le 2 \sum_{0 < j < m/2} e^{-\frac{n-k}{2} \bigl(\frac{j\pi}{m}\bigr)^2} \Bigl(\frac{j\pi}{m}\Bigr)^k .
\end{align*}
	We now observe that each term in the above summation is at most half of the previous one. This is because for every $j \ge 1$ we have
	\begin{align*}
	\frac{\exp\Bigl(-\frac{(n-k)(j+1)^2\pi^2}{2m^2} \Bigr) \Bigl(\frac{(j+1)\pi}{m}\Bigr)^k}{\exp\Bigl(-\frac{(n-k)j^2\pi^2}{2m^2} \Bigr) \Bigl(\frac{j\pi}{m}\Bigr)^k}
	&\le \exp\Bigl(-\frac{\pi^2(n-k)}{m^2} j\Bigr) \cdot \Bigl(1 + \frac{1}{j}\Bigr)^k \\
	&\le \exp\Bigl(-\frac{\pi^2(n-k)}{m^2}\Bigr) \cdot 2^k \\
	&\le \exp\Bigl(-\frac{3\pi^2}{4}k\Bigr) \cdot 2^k \\
	&\le \exp\Bigl(-\frac{3\pi^2}{4}\Bigr) \cdot 2 \\
	&\le \frac{1}{2},
	\end{align*}
  where the third step holds because w.l.o.g.\ $m \geq 2$ and $1 \le k \le n/m^2 \le n/4$, and hence $(n-k)/m^2 \ge 3k/4$.
	For these reasons, we can deduce that
	\[
	\sum_{j=0}^{m-1} \abs[\Big]{ \cos\frac{j\pi}{m} }^{n-k} \abs[\Big]{\sin\frac{j\pi}{m}}^k
	\le 2e^{-\frac{\pi^2(n-k)}{2m^2}} \Bigl(\frac{\pi}{m}\Bigr)^k \sum_{1 \le j < m/2} 2^{1-j}
	\le m^{-k} \exp\bigl(-\Omega(n/m^2)\bigr).
	\]

  For $k = 0$ we also need to include the term $p(\zeta^0) = p(1) = 2^n$ which divided by $m$ gives the term $2^n/m$.
\end{proof}

\subsection{Proof of \texorpdfstring{\Cref{lemma:bias-large-m}}{Lemma~\ref{lemma:bias-large-m}}} \label{sec:kravchuk}
In this section, we prove \Cref{lemma:bias-large-m}.
Let $\calK(n,k,w) := \sum_{x \in \calX_w} (-1)^{x_1 + \cdots + x_k}$ be the \emph{Kravchuk polynomial}.
We recall the ``symmetry relation'' for this polynomial (see e.g.~\cite{Kravchuk:wikipedia}):
\begin{align} \label{eq:kravchuk-identity}
\binom{n}{k} \calK(n,k,w)
= \sum_{y \in \calX_k} \sum_{x \in \calX_w} (-1)^{\angles{y,x}}
= \sum_{x \in \calX_w} \sum_{y \in \calX_k} (-1)^{\angles{x,y}}
= \binom{n}{w} \calK(n,w,k) .
\end{align}
Using \Cref{eq:kravchuk-identity}, we have
\begin{align}
\binom{n}{k} \max_{0 \le \ell \le m-1} \abs*{ \frac{\sum_{x \in \calX_{\ell \bmod m}}  (-1)^{x_1 + \cdots + x_k}}{\abs{\calX_{\ell \bmod m}}} }
&= \max_\ell\frac{\abs{\sum_{w \equiv \ell \bmod m} \binom{n}{k} \calK(n,k,w)}}{\sum_{w \equiv \ell \bmod m} \binom{n}{w}} \nonumber \\
&= \max_\ell\frac{\abs{\sum_{w \equiv \ell \bmod m} \binom{n}{w} \calK(n,w,k)}}{\sum_{w \equiv \ell \bmod m} \binom{n}{w}} \nonumber \\
&\le \max_\ell \frac{\sum_{c} \binom{n}{\floor{n/2}+cm+\ell} \abs{\calK(n,\floor{n/2}+cm+\ell,k)} }{\sum_{c} \binom{n}{\floor{n/2}+cm+\ell} } \nonumber \\
&\le 2 \max_\ell \frac{\sum_{c \ge 0} \binom{n}{\floor{n/2}+cm+\ell} \abs{\calK(n,\floor{n/2}+cm+\ell,k)} }{\sum_{c \ge 0} \binom{n}{\floor{n/2}+cm+\ell} } \label{eq:bias1},
\end{align}
where the indices of $c$ are integers that iterate over $\floor{n/2} + cm + \ell \in \{0, \ldots, n\}$, and the last line uses the symmetry of ${n \choose w} \calK(n,w,k)$ around  $w=n/2.$

Later in this section, we will prove (in \Cref{claim:kravchuk-bound}) that
\[
  \abs{\calK(n,\floor{n/2}+s,k)} \le 2^{O(k)} \left( \left(\frac{n}{k}\right)^{\frac{k}{2}} + \left(\frac{\abs{s}}{k}\right)^k \right) .
\]
Plugging this bound into \Cref{eq:bias1}, we have
\begin{align}
  \cref{eq:bias1}
  &\le 2^{O(k)} \max_\ell \left( \frac{\sum_{c\ge 0} \binom{n}{\floor{n/2}+cm+\ell} \bigl( (n/k)^{k/2} + ((cm+\ell)/k)^k \bigr) }{\sum_{c\ge 0} \binom{n}{\floor{n/2}+cm+\ell} } \right) \nonumber \\
  &= O(n/k)^{\frac{k}{2}} + \frac{2^{O(k)}}{k^k} \max_\ell \left( \frac{\sum_{c\ge 0} \binom{n}{\floor{n/2}+cm+\ell} (cm+\ell)^k }{\sum_{c\ge 0} \binom{n}{\floor{n/2}+cm+\ell} } \right) . \label{eq:bias2}
\end{align}
We will prove the following claim which bounds the ratio of the two summations.
\begin{claim} \label{claim:moment-mod-m}
  For any $\ell \in \{0, \ldots, m-1\}$ and $k \ge n/(2m)^2$, we have
  \[
    \frac{\sum_{c \ge 0} \binom{n}{\floor{n/2}+cm+\ell} (cm+\ell)^k }{\sum_{c \ge 0} \binom{n}{\floor{n/2}+cm+\ell} } 
    \le O(mk)^k .
  \]
\end{claim}
Therefore, as $(n/k)^{1/2} \le 2m$, we have
\[
  \cref{eq:bias2}
  \le O(n/k)^{\frac{k}{2}} + \frac{2^{O(k)}}{k^k} \cdot O(mk)^k 
  = O(m)^k .
\]
This completes the proof of \Cref{lemma:bias-large-m}. \qed

\medskip

It remains to prove \Cref{claim:moment-mod-m,claim:kravchuk-bound}. We begin with \Cref{claim:moment-mod-m}.

\begin{proof}[Proof of \Cref{claim:moment-mod-m}]
  Fix any $\ell \in \{0, \ldots, m-1\}$, we first bound
  \[
    \frac{\sum_{c \ge 0} \binom{n}{\floor{n/2}+cm+\ell} (cm+\ell)^k }{\sum_{c \ge 0} \binom{n}{\floor{n/2}+cm+\ell} } 
  \]
  with some quantity that does not depend on $\ell$.
  First, observe that the denominator is at least $\binom{n}{\floor{n/2} + \ell}$.
  Next, we have
  \begin{align*}
    \frac{\binom{n}{\floor{n/2}+\ell+cm}}{\binom{n}{\floor{n/2}+\ell}}
    &= \prod_{i=1}^{cm} \left(\frac{n - \floor{n/2} - \ell - cm + i}{\floor{n/2} + \ell + i} \right) \\
    &= \prod_{i=1}^{cm} \left(1 - \frac{2\ell+cm+\floor{n/2}-\ceil{n/2}}{\floor{n/2}+\ell+i} \right) \\
    &\le \left(1 - \frac{2\ell+cm-1}{n} \right)^{cm} \\
    &\le \exp\left(-\frac{cm(cm-1)}{n}\right) \\
    &\le e \cdot \exp(-(cm)^2/n) && (e^{cm/n} \le e) \\
    &\le e \cdot \exp(-(c^2/4k)) && (k \ge n/(2m)^2) .
  \end{align*}
  Finally, we have $\ell \leq m$ and $(cm+\ell) \le (c+1) m \le 2cm$ for $c \ge 1$.
  Therefore
  \begin{align*}
    \frac{\sum_{c \ge 0} (cm+\ell)^k \binom{n}{\floor{n/2}+cm+\ell} }{\sum_{c \ge 0} \binom{n}{\floor{n/2}+cm+\ell} }
    &\le \frac{\sum_{c \ge 0} (cm+\ell)^k \binom{n}{\floor{n/2}+cm+\ell} }{\binom{n}{\floor{n/2}+\ell} }  \\
    &\le m^k + e \sum_{c \ge 1} (2cm)^k e^{-\frac{c^2}{4k}} \\
    &=  m^k + e (2m)^k \sum_{c \ge 1} c^k e^{-\frac{c^2}{4k}} .
  \end{align*}
  In \Cref{claim:moment} below we will show that the summation in the last line is bounded by $O(k^k)$.
  Applying this bound completes the proof.
\end{proof}
\begin{claim} \label{claim:moment}
  $\sum_{c \ge 1} c^k e^{-\frac{c^2}{4k}} \le O(k^k)$ for any $k \ge 1$.
\end{claim}
\begin{proof}
Consider the function
\[
  \lambda_a(x) := x^{k} \cdot e^{-a x^{2}},
\]
for some parameter $a > 0$.  Its derivative is $\lambda_a'(x) = (k-2ax^2) x^{k-1} e^{-ax^2}$; so for $x \geq 0$, $\lambda_a(x)$ is increasing when $x \leq \sqrt{k / (2a)}$, and is decreasing when $x \geq \sqrt{k / (2a)}$. For such a (nonnegative) {\em first-increasing-then-decreasing} function, it can be seen that
\begin{align*}
\sum_{c \geq 1} \lambda_{a}(c)
& \leq \underbrace{\vphantom{\int_0^\infty}\max \{\lambda_{a}(x): x \geq 0\}}_{\term[A]{term:T1}} + \underbrace{\int_{0}^{\infty} \lambda_{a}(x) dx}_{\term[B]{term:T2}}.
\end{align*}
The first term \cref{term:T1} is exactly equal to
\[
\cref{term:T1} = \lambda_{a}\left(\sqrt{k / (2a)}\right) = \left(\frac{k}{2e a}\right)^{k / 2} .
\]
Using the substitution $u=ax^2$, the second term \cref{term:T2} is equal to
\[
\cref{term:T2}
= \int_{0}^{\infty} x^k e^{-a x^{2}} dx
= \frac{1}{2} a^{-\frac{k + 1}{2}} \cdot \Gamma\left(\frac{k + 1}{2}\right) .
\]
where $\Gamma(k) := \int_0^\infty x^{k-1}e^{-x} dx$ is the Gamma function.
It is known (see, for example, \cite[Equation~(5.6.1)]{OLBC10}) that $\Gamma(\frac{k + 1}{2}) \leq \sqrt{2\pi} \cdot e^{\frac{1}{12}} \cdot (\frac{k}{2e})^{k / 2}$ for each integer $k \geq 1$. Plugging this upper bound into the above formula gives
\[
\cref{term:T2}
\leq e^{\frac{1}{12}} \cdot \sqrt{\frac{\pi}{2a}} \cdot \left(\frac{k}{2e a}\right)^{k / 2}
\leq \frac{2}{\sqrt{a}} \cdot \left(\frac{k}{2e a}\right)^{k / 2}.
\]
By setting $a = 1/(4k)$, we deduce that
\[
  \cref{term:T1} + \cref{term:T2}
  \leq \bigl(1 + 4\sqrt{k} \bigr) \cdot \left( 2 k^2/e \right)^{k / 2}
  = \frac{1 + 4\sqrt{k}}{(e / 2)^{k / 2}} \cdot k^{k}
  = O(1) \cdot k^k ,
\]
proving the claim.
\end{proof}

It remains to prove \Cref{claim:kravchuk-bound}, which gives a bound on $\calK(n,w,k) := \sum_{x \in \calX_k} (-1)^{x_1 + \cdots + x_w}$.
\begin{claim} \label{claim:kravchuk-bound}
  $\abs{\calK(n,\floor{n/2}+s,k)} \le 2^{O(k)} ((n/k)^{k/2} + (\abs{s}/k)^k)$ for any $k \geq 1.$
\end{claim}
A similar claim was also made in \cite[Claim~4.10]{CohenPT21}, but the bounds we obtain here are more general and sharper for larger $\abs{s}$, which is crucial for our application.
We start with a preliminary claim.
\begin{claim} \label{claim:binomial}
 For any positive integer $s$,
 \[
   \sum_{i=0}^{\floor{k/2}} \frac{n^i}{i!} \frac{s^{k-2i}}{(k-2i)!}
   \le \begin{cases}
     (2e)^k \left(\frac{n}{k}\right)^{k/2} & \text{if $s \le \sqrt{nk}$} \\
     (2e)^k \left(\frac{s}{k}\right)^k & \text{if $s \ge \sqrt{nk}$.}
   \end{cases}
 \]
\end{claim}
\begin{proof}
  Write $s = (cn)^{1/2}$ for some $c > 0$.
  Then $n^i s^{k-2i} = (cn)^{k/2}/c^i$, and
  \begin{align}
   \sum_{i=0}^{\floor{k/2}} \frac{n^i}{i!} \frac{s^{k-2i}}{(k-2i)!}
   = (cn)^{\frac{k}{2}} \sum_{i=0}^{\floor{k/2}} \frac{1}{c^i i! (k-2i)!} . \label{eq:binom1}
  \end{align}
We can write 
\[
  \frac{1}{i!(k-2i)!}
  = \frac{(2i)!}{k! i!} \frac{k!}{(2i)!(k-2i)!}
  = \frac{(2i)!}{k! i!} \binom{k}{2i} .
\]
So
\begin{align*}
  \cref{eq:binom1}
  &= \frac{(cn)^{\frac{k}{2}}}{k!} \sum_{i=0}^{\floor{k/2}} \frac{(2i)!}{c^i i!} \binom{k}{2i} \\
  &\le \frac{(cn)^{\frac{k}{2}}}{k!} \max_{0 \le i \le \floor{\frac{k}{2}}} \left(\frac{2i}{c}\right)^i \cdot \sum_{i=0}^{\floor{k/2}} \binom{k}{2i} \\
  &\le 2^k \cdot \frac{(cn)^{\frac{k}{2}}}{k!} \max_{0 \le i \le \floor{\frac{k}{2}}} \left(\frac{2i}{c}\right)^i . 
\end{align*}
If $c \le k$, then for $0 \le i \le \floor{k/2}$, we have $2i/c \le k/c$ and thus $(2i/c)^i \le (k/c)^i \le (k/c)^{k/2}$ because $k/c \ge 1$.
Hence,
\[
  2^k \cdot \frac{(cn)^{\frac{k}{2}}}{k!} \max_{0 \le i \le \floor{\frac{k}{2}}} \left(\frac{2i}{c}\right)^i
  \le 2^k \frac{(cn)^{\frac{k}{2}}}{k!} \left(\frac{k}{c}\right)^{\frac{k}{2}}
  \le (2e)^k \left(\frac{n}{k}\right)^{k/2} .
\]
If $c \ge k$, then for $0 \le i \le \floor{k/2}$, we have $(2i/c)^i \le 1$, and thus
\[
  2^k \cdot \frac{(cn)^{\frac{k}{2}}}{k!} \max_{0 \le i \le \floor{\frac{k}{2}}} \left(\frac{2i}{c}\right)^i
  \le (2e)^{k} \left(\frac{cn}{k^2}\right)^{\frac{k}{2}}
  = (2e)^{k} \left(\frac{s}{k}\right)^k . \qedhere
\]
\end{proof}

\begin{proof}[Proof of \Cref{claim:kravchuk-bound}]
	Observe that
	\[
	(1-x)^w (1+x)^{n-w} = \sum_{k=0}^n \calK(n,w,k) x^k,
	\]
	Our goal is to bound above its $k$-th coefficient.
	We will assume $n$ is even, otherwise we can apply the bound for $n+1$, which affects its magnitude by at most a factor of $2$.
	By the symmetry of $\calK(n,u,k)$ around $u=n/2$, for \Cref{claim:kravchuk-bound} we may take $w = n/2-s$, and for this choice of $w$ we have
	\begin{align} \label{eq:kravchuk1}
	\sum_{k=0}^n \calK(n,n/2-s,k) x^k
	= (1-x)^{\frac{n}{2}-s} (1+x)^{\frac{n}{2}+s}
	= (1-x^2)^{\frac{n}{2}} \left(\frac{1+x}{1-x}\right)^s .
	\end{align}
	We bound from above the coefficients of the power series expansion of the two terms on the right hand side.
  The $i$-th coefficient of $(1-x^2)^{n/2}$ is $\binom{n/2}{i/2}$ if $i$ is even and is $0$ otherwise (because the function is even).
  For $(\frac{1+x}{1-x})^s$, we will assume that $s$ is nonnegative as this only affects the sign of its coefficients.
  Since $\frac{1+x}{1-x} = 1 + \sum_{j \ge 1} 2x^j$ and $\frac{1}{1-2x} = 1 + \sum_{j \ge 1} (2x)^j$ for $\abs{x} < 1/2$ in their power series expansion, the $j$-th coefficient of $(\frac{1+x}{1-x})^s$ is at most the $j$-th coefficient of $(1-2x)^{-s}$, which is
		\[
		2^j \cdot \frac{s \cdot (s+1) \cdots (s+j-1)}{1 \cdot 2 \cdots j}
		= \binom{s+j-1}{j} 2^j .
		\]
	Therefore, the magnitude of the $k$-th coefficient of \Cref{eq:kravchuk1} is at most
	\begin{align} \label{eq:kravchuk2}
	\sum_{i=0}^{\floor{k/2}} \binom{n/2}{i} \binom{\abs{s}+k-2i-1}{k-2i} 2^{k-2i} .
	\end{align}
	Using $\binom{n}{k} \le n^k/k!\,$, we have
	\begin{align*}
    \cref{eq:kravchuk2}
    \le \sum_{i=0}^{\floor{k/2}} \frac{n^i}{i!} \frac{(\abs{s}+k-2i-1)^{k-2i}}{(k-2i)!} 2^{k-2i} .
	\end{align*}
  Note that $(s+t)^j \le 2^j \max\{s^j, t^j\}$ and so
	\[
    \frac{(\abs{s}+k-2i-1)^{k-2i}}{(k-2i)!}
    \le 2^{k-2i} \left( \frac{\max\{\abs{s}^{k-2i}, (k-2i-1)^{k-2i}\}}{(k-2i)!}\right)
    \le 2^{k-2i} \left( \frac{\max\{\abs{s}^{k-2i}, k^{k-2i}\}}{(k-2i)!}\right) .
	\]
	Therefore, \cref{eq:kravchuk2} is bounded by
	\[
    2^{2k} \left( \sum_{i=0}^{\floor{k/2}} \frac{n^i}{i!} \frac{\max\{\abs{s}^{k-2i}, k^{k-2i}\}}{(k-2i)!} \right) .
  \]
  The claim then follows from \Cref{claim:binomial}.
\end{proof}

\subsection{Proof of \Cref{thm:sym-lb}} \label{sec:sym-lb}
In this subsection, we give examples of symmetric polynomials $p$ with $L_{1,k}(p) = \Omega_k(1) \cdot d^k$ for $d = \Theta(\sqrt{n})$, matching our upper bound up to a constant factor when $k = O(1)$.
We restate our theorem for convenience.

\begin{theorem}[Restatement of \Cref{thm:sym-lb}] \label{thm:sym-lb-restate}
  For every $k$, there is a symmetric $\F_2$-polynomial $p(x_1, \ldots, x_n)$ of degree $d = \Theta(\sqrt{kn})$ such that $L_{1,k}(p) \ge (e^{-k}/2) \cdot \binom{n}{k}^{1/2} = \Omega_k(1) \cdot d^k$.
\end{theorem}
\begin{proof}
  Let $m := C \sqrt{kn}$ for some constant $4 \le C \le 8$ such that $m$ is a power of two.
  Consider the function $h_k\colon\zo^n \to \pmone$ defined by
  \[
    h_k(x) = \sgn\Bigg(\sum_{\abs{S}=k} (-1)^{\sum_{i\in S} x_i} \Bigg) .
  \]
  Let $\widetilde{h_k}\colon \{0,\ldots,n\} \to \pmone$ be the symmetrization of $h_k$ so that $\widetilde{h_k}(\abs{x}) := h_k(x)$.
  Now consider the periodic function $\widetilde{p}\colon \{0, \ldots, n\} \to \pmone$ that on input $z \in \{0, \ldots, n\}$ such that
  \[
    z + 2tm \in [n/2 - m, n/2 + m) \text{ for some $t \in \Z$},
  \]
  evaluates to $\widetilde{p}(z) := \widetilde h_k(z+2tm)$.
  We define the $\F_2$-polynomial $p$ by $p(x_1, \ldots, x_n) := \tilde{p}(\abs{x})$.
  Since $p(x)$ depends only on $\abs{x} \bmod 2m$, by \Cref{lemma:BGL}, it can be computed by a symmetric $\F_2$-polynomial of degree  $d < 2m$.
  Moreover, $p(x)$ agrees with $h_k(x)$ when $\abs{x} \in [n/2-m, n/2+m)$.
	Let $f = (-1)^p$.
  We have (note that from here on we switch from $x\in \zo^n$ to $x \in \pmone^n$)
	\begin{align*}
    L_{1,k}(f)
     \ge \abs[\Bigg]{\sum_{\abs{S}=k} \hf(S)}
    &= \abs[\Bigg]{ \sum_{\abs{S}=k} \E_x \Bigl[ f(x) x^S \Bigr]} \\
    &= \abs[\Bigg]{ \sum_{\abs{S}=k} \E_x\Bigl[h_k(x) x^S\Bigr] + \E_x\Bigl[\bigl(f(x) - h_k(x) \bigr) x^S \Bigr] } \\
    &= \abs[\Bigg]{ \E_x\Biggl[h_k(x) \sum_{\abs{S}=k} x^S\Biggr] + \E_x\Biggl[\bigl(f(x) - h_k(x) \bigr) \sum_{\abs{S}=k}  x^S \Biggr] } \\
    &\ge \left| \abs[\Bigg]{ \E_x\Biggl[h_k(x) \sum_{\abs{S}=k} x^S\Biggr] } - \abs[\Bigg]{ \E_x\Biggl[\bigl(f(x) - h_k(x) \bigr) \sum_{\abs{S}=k}  x^S \Biggr] } \right| .
	\end{align*}
  We proceed to lower bound the first term and upper bound the second term.
  For the first term we use the same lower bound of $e^{-k}\binom{n}{k}^{1/2}$ in the proof of \cite[Theorem~1]{Viola20};
  for completeness we include the argument here.
  One can verify that $\E_x[(\sum_{\abs{S}=k}  x^S)^2] = \binom{n}{k}$.
  Applying \cite[Theorem~9.22]{ODonnell:book}, which states that $\E[\abs{h(x)}] \ge e^{-k} \E[h(x)^2]^{1/2}$ for any $h\colon\pmone^n \to \R$ of degree $k$, we get
  \begin{align*}
    e^{-k} \binom{n}{k}^{1/2}
    = e^{-k} \E_x\Biggl[ \Biggl(\sum_{\abs{S}=k}  x^S \Biggr)^2 \Biggr]^{1/2}
    \le \E_x\Biggl[ \abs[\Bigg]{ \sum_{\abs{S}=k}  x^S } \Biggr]
    \le \E_x\Biggl[ \Biggl(\sum_{\abs{S}=k}  x^S \Biggr)^2 \Biggr]^{1/2}
    = \binom{n}{k}^{1/2} .
  \end{align*}
  So for the first term we have
  \[
    \E_x \Biggl[h_k(x) \sum_{\abs{S}=k} x^S\Biggr]
    = \E_x \left[\abs[\Bigg]{ \sum_{\abs{S}=k}  x^S } \right]
    \ge e^{-k} \binom{n}{k}^{1/2} .
  \]
  We now bound above the second term.
  As $\abs{f(x) - h_k(x)} \le 2$ for every $x$, by the Cauchy--Schwarz inequality, we get
  \begin{align*}
    \left| \E_x\Biggl[\bigl(f(x) - h_k(x) \bigr) \sum_{\abs{S}=k}  x^S \Biggr] \right|
    &\le 2 \cdot \Pr_x\biggl[ 2\abs[\Big]{\sum_i x_i} \ge m \biggr] \E_x\left[ \Biggl( \sum_{\abs{S}=k} x^S \Biggr)^2 \right]^{1/2} \\
    &= 2 \cdot \Pr_x\biggl[ 2\abs[\Big]{\sum_i x_i} \ge m \biggr] \binom{n}{k}^{1/2} \\
    &\le (e^{-k}/2) \cdot \binom{n}{k}^{1/2} ,
\end{align*}
  where the last inequality follows from the Chernoff bound recalling our choice of $m = \Theta(\sqrt{kn})$.
  Therefore $L_{1,k}(f) \ge (e^{-k}/2) \cdot \binom{n}{k}^{1/2}$, proving the theorem.
\end{proof}


\section{$L_{1,k}$ bounds for read-$\Delta$ polynomials} \label{sec:read-few}

In this section, we give our $L_{1,k}$ bounds on read-few polynomials, which are restated below for convenience.

\begin{theorem}[Restatement of \Cref{thm:read-few-poly}] \label{thm:read-few}
  Let $p(x_1, \ldots, x_n)$ be any read-$\Delta$ degree-$d$ polynomial.
  For every $1 \le k \le n$,
  \[
    L_{1,k}(p) \le \Pr[p=1] \cdot O(k)^k \cdot (\Delta d)^{10k} .
  \]
\end{theorem}

By observing that any degree-$\Omega(\log n)$ monomial vanishes under a random restriction with high probability, we can use \Cref{thm:read-few} to obtain an upper bound for read-$\Delta$ polynomials that is independent of $d$:

\begin{corollary}  [Restatement of \Cref{cor:L1k-read-few-poly}] \label{cor:L1k-read-few}
  Let $p(x_1, \ldots, x_n)$ be a read-$\Delta$ polynomial.
  For any $1 \le k \le n$,
  \[
    L_{1,k}(p)
    \le O(k)^{11k} \cdot (\Delta \log n)^{10k}  .
  \]
\end{corollary}
\begin{proof}
  Let $R$ be a $\rho$-random restriction that independently for each $x_i$ with probability $1-\rho$ sets  $x_i$ to a uniform random bit and keeps $x_i$ alive with the remaining $\rho$ probability.
  For any $f\colon \pmone^n \to \pmone$, we have $\E_R[\widehat{f_R}(S)] = \rho^{\abs{S}} \hf(S)$.
  Thus,
  \[
    L_{1,k}(f)
    = \sum_{\abs{S}=k} \abs{\hf(S)}
    = \sum_{\abs{S}=k} \abs{\rho^{-k} \E_R[\widehat{f_R}(S)] }
    \le \rho^{-k} \E_R[L_{1,k}(f_R)] .
  \]
  Let $p(x_1, \ldots, x_n)$ be any read-$\Delta$ polynomial.
  Note that $L_{1,k}(f) \le n^{k/2}$ for any $f$.
  Hence, by \Cref{thm:read-few}, for any $d_{\max}$ and any $\rho$, we have
  \[
    L_{1,k}(p)
    \le \rho^{-k} \cdot O(k)^k (\Delta d_{\max})^{10k}  + \rho^{-k} \cdot \Pr_R[\deg(p_R) \ge d_{\max}] \cdot n^{k/2} .
  \]
  We now upper bound $\Pr_R[\deg(p_R) \ge d_{\max}]$.
  Each monomial with degree greater than $d_{\max}$ survives $R$ with probability at most $\rho^{d_{\max}}$, and there are at most $(\Delta n)/d_{\max}$ such monomials in $p$.
  It follows by a union bound that $\Pr_R[\deg(p_R) \ge d_{\max}] \le \rho^{d_{\max}}(\Delta n)/d_{\max}$.
  So setting $\rho = 1/2$ and $d_{\max} = 2k\log n$, we have
  \begin{align*}
    L_{1,k}(p)
    &\le O(k)^k (\Delta \cdot d_{\max})^{10k}  + \frac{\rho^{d_{\max}} \Delta \cdot n^{1+k/2}}{\rho^k d_{\max}} \\
    &\le O(k)^{11k} \cdot (\Delta \log n)^{10k} . \qedhere
  \end{align*}

\end{proof}

\paragraph{Proof idea.}
We first observe that for $f = (-1)^p$, the Fourier coefficient $\hf(S)$ is simply the bias of the $\F_2$-polynomial $p_S(x) := p(x) + \sum_{i \in S} x_i$.
Assuming that $p_S$ depends on all $n$ variables, by a simple greedy argument we can collect $n/\poly(\Delta,d)$ polynomials in $p_S$ so that each of them depends on disjoint variables, and it is not hard to show that the product of the biases of these polynomials upper bounds the bias of $p_S$. 
From this it is easy to see that any read-$\Delta$ degree-$d$ polynomial has bias $\exp(2^{-d}n/\poly(\Delta,d))$.
However, this quantity is too large to sum over $\binom{n}{k}$ coefficients.

Our next idea (\Cref{lemma:read-few-key}) is to give a more refined decomposition of the polynomial $p$ by inspecting the variables $x_i: i \in S$ more closely.
Suppose the variables $x_i : i \in S$ are far apart in their dependency graph (see the definition of $G_p$ below), as must indeed be the case for most of the ${n \choose k}$ size-$k$ sets $S$.
Then we can collect all the monomials containing each $x_i$ to form a polynomial $p_i$, and these $p_i$'s will depend on disjoint variables.
Moreover, if every monomial in $p_i$ has high degree (see the definition of $V_t(p)$ below), then $p_i = 0$ with high probability and therefore $p_i + x_i$ is almost unbiased.
Therefore, we can first collect these $p_i$ and $x_i$ from $p_S$; then, for the remaining $m \ge \abs{S} \cdot \poly(\Delta,d)$ monomials in $p_S$, as before we collect $m/\poly(\Delta,d)$ polynomials $r_i$ so that they depend on disjoint variables, but this time we collect these monomials using the variables in $V_t(p)$, and give an upper bound in terms of the size $\abs{V_t(p)}$.
Multiplying the biases of the $p_i + x_i$'s and the bias of $r$ gives our refined upper bound on $\hf(S)$ in \Cref{lemma:read-few-key}.

\medskip

We now proceed to the actual proof.
We first define some notions that will be used throughout our arguments.
For a read-$\Delta$ degree-$d$ polynomial $p$, we define $V_t(p): t \in [d]$ and $G_p$ as follows.

For every $t \in [d]$, define 
\[
  V_t(p) := \{i \in [n]: \text{the minimum degree of the monomials in $p$ containing $x_i$ is $t$} \} .
\]
Note that the sets $V_1(p),\dots,V_d(p)$ form a partition of the input variables $p$ depends on.

Define the undirected graph $G_p$ on $[n]$, where $i,j \in [n]$ are adjacent if $x_i$ and $x_j$ both appear in the same monomial in $p$.
Note that $G_p$ has degree at most $\Delta d$.
For $S \subseteq [n]$, we use $N_{=d}(S)$ to denote the indices that are at distance exactly $i$ to $S$ in $G_p$, and use $N_{\le d}(S)$ to denote $\bigcup_{j=0}^d N_{=j}(S)$.

We first state our key lemma, which gives a refined bound on each $\hf(S)$ stronger than the naive bound sketched in the first paragraph of the ``Proof Idea'' above, and use it to prove \Cref{thm:read-few}.
We defer its proof to the next section.

\begin{lemma} [Main lemma for read-$\Delta$ polynomials] \label{lemma:read-few-key}
  Let $p(x_1, \ldots, x_n)$ be a read-$\Delta$ degree-$d$ polynomial.
  Let $S \subseteq [n], \abs{S} \ge \ell$ be a subset containing some $\ell$ indices $i_1, \ldots, i_\ell \in S$ whose pairwise distances in $G_p$ are at least $6$, and let $t_1,\dots,t_\ell \in [d]$ be such that each $i_j \in V_{t_j}(p)$.
  Let $f = (-1)^p$.
  Then 
  \[
    \abs{\hf(S)}
    \le O(1)^{\abs{S}} \cdot \Delta^\ell \prod_{j \in [\ell]} \left( 2^{-t_j} \exp\left(-\frac{2^{-{t_j}} \abs{V_{t_j}(p)}}{\ell \cdot (\Delta d)^4} \right) \right) .
  \]
\end{lemma}

\begin{proof}[Proof of \Cref{thm:read-few}]
  Using a reduction given in the proof of \cite[Lemma 2.2]{ChakrabortyMMMPS20}, it suffices to prove the same bound without the acceptance probability factor, i.e.~to prove that for every $1 \le k\le n$,
  \[
    L_{1,k}(p) \le O(k)^k \cdot (\Delta d)^{10k} .
  \]
  As \cite{ChakrabortyMMMPS20} did not provide an explicit statement of the reduction, for completeness we provide a self-contained statement and proof in \Cref{lemma:xor-reduction-restate} in \Cref{sec:reduction}.

  For every subset $S \subseteq [n]$ of size $k$, there exists an $\ell \le k$ and $i_1, \ldots, i_\ell \in S$ such that their pairwise distances in $G_p$ are at least $6$, each $i_j \in V_{t_j}(p)$ for some $t_j \in [d]$, and each of the remaining $k - \ell$ indices in $S$ is within distance at most $5$ to some $i_j$.
  
  Fix any $i_1, \ldots, i_\ell$, and let us bound the number of subsets $S \subseteq [n]$ of size $k$ that can contain  $i_1, \ldots, i_\ell$. Because $\abs{N_{\le 5}(j)} \le \sum_{i=0}^5 (\Delta d)^i \leq 6(\Delta d)^5$ for every $j \in [n]$, the remaining $k-\ell$ indices of $S$ can appear in at most
  \begin{align*}
    \sum_{j_1 + \cdots + j_\ell = k-\ell} \prod_{b \in [\ell]} \binom{6(\Delta d)^5}{j_b}
    &= \binom{6 \ell (\Delta d)^5} {k-\ell} \\
    &\le (6(\Delta d)^5)^k \cdot e^{k-\ell} \left(\frac{\ell}{k-\ell}\right)^{k-\ell} \\
    &\le (e \Delta d)^{5k}
  \end{align*}
  different ways, where the equality uses the Vandermonde identity, the first inequality uses $\binom{n}{k} \le (en/k)^k$, and the last one uses $(\frac{\ell}{k-\ell})^{k-\ell} \le (1 + \frac{\ell}{k-\ell})^{k-\ell} \le e^\ell$ and $6e < e^5$.
  Therefore, by \Cref{lemma:read-few-key},  
\begin{align*}
    \sum_{S: \abs{S}=k} \abs{\hf(S)}
&\le \sum_{\ell=1}^k \sum_{t \subseteq [d]^\ell} \left[ \Biggl( \prod_{j \in [\ell]} \abs{V_{t_j}(p)} \Biggr) \cdot (e \Delta d)^{5k} \cdot O(1)^k 
\Delta^\ell \prod_{j' \in [\ell]} \left( 2^{-t_{j'}} \exp\left(-\frac{2^{-{t_{j'}}} \abs{V_{t_{j'}}(p)}}{\ell (\Delta d)^4} \right) \right) \right] \\
    &\le O(1)^k \cdot (\Delta d)^{5k} \sum_{\ell=1}^k \Delta^\ell \sum_{t \subseteq [d]^\ell} \prod_{j \in [\ell]} \left( 2^{-t_j} \abs{V_{t_j}(p)} \exp\left(-\frac{2^{-{t_j}} \abs{V_{t_j}(p)}}{\ell (\Delta d)^4} \right) \right) \\
    &\le O(1)^k \cdot (\Delta d)^{5k} \sum_{\ell=1}^k \Delta^\ell \cdot d^\ell \cdot (\ell (\Delta d)^4 )^\ell \\
    &\le O(k)^k \cdot (\Delta d)^{5k} \cdot (\Delta d)^{5k} \\
    &= O(k)^k \cdot (\Delta d)^{10k} ,
  \end{align*}
  where the third inequality is because the function $x \mapsto xe^{-x/c}$ is maximized when $x=c$.
  This completes the proof.
\end{proof}

\subsection{Proof of \Cref{lemma:read-few-key}}
  First we need a simple lemma that lower bounds the acceptance probability of a sparse polynomial. 

  \begin{lemma} \label{lemma:helper}
    Let $r(x_1, \ldots, x_n)$ be a polynomial with at most $\Delta$ monomials, each of which has degree at least $t$, and $r(0) = 0$.
    Then $\Pr[r(x) = 0] \geq \max\{\frac{1}{\Delta + 1}, 1-2^{-t} \Delta\}$.
  \end{lemma}
  \begin{proof}
    The lower bound of $1/(\Delta+1)$ follows from \cite[Corollary 1]{KL93}.
    For the other bound, note that each monomial evaluates to $1$ with probability at most $2^{-t}$; so by a union bound $\Pr[r(x) = 1] \le 2^{-t} \Delta$.
  \end{proof}

  We use $G_p$ to partition $[n]$ as follows.
  For each $j \in [\ell]$, let $T_j := N_{=1}(i_j)$, the set of distance-1 neighbors of $i_j$, and $U_j := N_{=2}(i_j)$, the set of distance-2 neighbors of $i_j$.
  Let $R := [n] \setminus \bigcup_j  (\{i_j\} \cup T_j \cup U_j)$ be the remaining variables (note that these are the variables whose distance from every $i_j$ is at least 3).
  As the $i_j$'s are at distance at least $6$ apart, then for $j \neq j'$ each index in $\{i_j\} \cup T_j \cup U_j$ has distance at least $2$ to each index in  $\{i_{j'}\} \cup T_{j'} \cup U_{j'}$; so the variables in $\{i_j\} \cup T_j \cup U_j$ appear in disjoint monomials from the variables in $\{i_{j'}\} \cup T_{j'} \cup U_{j'}$.
  We now partition the monomials in $p$.
  For each $j \in [\ell]$, let $p_j, q_j, v_j$ be 3 polynomials such that $p_j$ is composed of all the monomials in $p$ that contain $x_{i_j}$, $q_j$ is composed of all monomials that contain some variable in $x_{T_j}$ but not $x_{i_j}$, and $w_j$ is composed of all monomials that contain some variable in $x_{U_j}$, but not $x_{i_j}$ nor any variable in $x_{T_j}$.
  Let $r$ be the polynomial composed of the remaining monomials in $p$.
  We can write $p$ as
  \[
    \sum_{j=1}^{\ell} \bigl( p_j(x_{i_j},x_{T_j}) + q_j(x_{T_j}, x_{U_j}) + w_j(x_{U_j}, x_R) \bigr) + r(x_R) .
  \]
  Now, consider the polynomial
  \begin{align}
    p_S(x_1, \ldots, x_n)
    &:= p(x_1, \ldots, x_n) + \sum_{i\in S} x_i \nonumber \\
    &\phantom{:}= \sum_{j=1}^\ell s_j(x_{i_j}, x_{T_j}, x_{U_j}) + r'(x_{U_1}, \ldots, x_{U_\ell}, x_R), \label{eq:tomato}
  \end{align}
  where
  \begin{align}
    s_j(x_{i_j}, x_{T_j}, x_{U_j})
    &:= p_j(x_{i_j},x_{T_j})  + q_j(x_{T_j}, x_{U_j}) + x_{i_j} + \sum_{k \in S \cap T_j} x_k \label{eq:pepper}\\
    r'(x_{U_1}, \ldots, x_{U_\ell}, x_R)
    &:= \sum_{j=1}^\ell \biggl( w_j(x_{U_j},x_R) + \sum_{k \in S \cap U_j} x_k \biggr) + r(x_R) + \sum_{k \in S \cap R} x_k .\nonumber
  \end{align}

  The following claim gives an upper bound on the bias of $s_j$; we defer its proof until later.
  (To interpret the claim it may be helpful to recall that $i_j \in V_{t_j}(p)$, and hence the minimum degree of the monomials in $p$ containing $x_{i_j}$ is $t_j$.)
  \begin{claim} \label{claim:xor-high-deg-monomials}
    For every $j \in [\ell]$ and each possible outcome of $x_{U_j}: j \in [\ell]$, we have $\abs{\E_{x_{i_j}, x_{T_j} }[ (-1)^{s_j(x_{i_j}, x_{T_j}, x_{U_j}})} ] \le \Delta 2^{-(t_j-1)}$.
  \end{claim}
We have that  
  \begin{align}
    \abs{\hf(S)}
    &= \abs[\bigg]{ \E\Bigl[ (-1)^{p_S(x)} \Bigr] } \nonumber\\
    &= \abs[\Bigg]{ \E_{x_{U_j}:j\in [\ell]} \biggl[ \prod_{j \in [\ell]} \E_{x_{T_j}, x_{i_j}} \Bigl[(-1)^{s_j(x_{i_j}, x_{T_j}, x_{U_j})} \Bigr] \cdot \E_{x_R} \Bigl[(-1)^{r'(x_{U_1}, \ldots, x_{U_\ell}, x_R)} \Bigr] } \biggr] \nonumber \\
    &\le  \E_{x_{U_j}:j\in [\ell]} \Biggl[ \abs[\bigg]{ \prod_{j \in [\ell]} \E_{x_{T_j}, x_{i_j}} \Bigl[(-1)^{s_j(x_{i_j}, x_{T_j}, x_{U_j})} \Bigr] } \cdot \E_{x_R} \abs[\bigg]{ \Bigl[(-1)^{r'(x_{U_1}, \ldots, x_{U_\ell}, x_R)} \Bigr] } \Biggr] \label{eq:a}\\
     &\le (2\Delta)^\ell \cdot \biggl(\prod_{j \in [\ell]} 2^{-t_j}\biggr) \cdot \E_{x_{U_j}:j\in[\ell]} \Biggl[ \abs[\bigg]{\E_{x_R}  \Bigl[ (-1)^{r'(x_{U_1}, \ldots, x_{U_\ell}, x_R)} \Bigr] } \Biggr] , \label{eq:b}
  \end{align}
  where \Cref{eq:a} is by independence of the different $(x_{i_j}, x_{T_j})$'s and $x_R$, and \Cref{eq:b} is by \Cref{claim:xor-high-deg-monomials}.
  For every fixed assignment $u = (u_1, \ldots, u_\ell)$ to the variables in $x_{U_1}, \ldots, x_{U_\ell}$, let $r'_u(x_R)$ denote the restricted polynomial $r'(u, x_R)$.
  We now claim that for every $u$,
  \begin{equation} \label{eq:dill}
    \abs{V_t(r'_u)} \ge \abs{V_t(p)} - \ell(\Delta d)^3 - \abs{S} \quad\text{for every $t\in [d]$.}
  \end{equation}
  This is because a variable $x_j$ belongs to $V_t(p) \setminus V_t(r'_u)$ only if it belongs to or is adjacent to the at most $\ell (\Delta d)^2$ restricted variables $\{x_{i_j}, x_{T_j}, x_{U_j}: j \in [\ell]\}$ in $G_p$, or it is one of the $\{x_i: i \in S\}$.
  In \Cref{lemma:disjoint-decomp} below we prove that for every fixed assignment $u$ to $x_{U_j}: j\in [\ell]$, and any $t \in [d]$, we have that
  \begin{equation} \label{eq:gherkin}
  \abs[\bigg]{ \E_{x_R} \Bigl[(-1)^{r'(x_{U_1}, \ldots, x_{U_\ell}, x_R)} \Bigr] }
    \le \exp\left(- \frac{2^{-t}\abs{V_t(r'_u)}}{(\Delta d)^4} \right).
  \end{equation}
  Continuing from above, combining \Cref{eq:b,,eq:dill,eq:gherkin}, we have
  \begin{align*}
    \abs{\hf(S)}
    &\le (2\Delta)^\ell \Biggl(\prod_{j \in [\ell]} 2^{-t_j}\Biggr) \cdot \min \left\{\exp\left(-\frac{2^{-t} (\abs{V_t(p)} - \ell (\Delta d)^3 - \abs{S})}{(\Delta d)^4} \right) : t\in [d] \right\}  \\
    &\le O(1)^{\abs{S}} \cdot (2\Delta)^\ell \Biggl(\prod_{j \in [\ell]} 2^{-t_j} \Biggr) \cdot \min \left\{\exp\left(-\frac{2^{-t} \cdot \abs{V_t(p)}}{(\Delta d)^4} \right) : t\in [d] \right\}  \\
    &= O(1)^{\abs{S}} \cdot (2\Delta)^\ell \Biggl( \prod_{j \in [\ell]} 2^{-t_j} \min \left\{\exp\left(-\frac{2^{-t} \cdot \abs{V_t(p)}}{\ell \cdot (\Delta d)^4} \right) : t\in [d] \right\} \Biggr)  \\
    &= O(1)^{\abs{S}} \cdot (2\Delta)^\ell \prod_{j \in [\ell]} \left( 2^{-t_j} \exp\left(-\frac{2^{-{t_j}} \cdot \abs{V_{t_j}(p)}}{\ell \cdot (\Delta d)^4} \right) \right) ,
  \end{align*}
  where second inequality is because $\exp\left(\frac{\ell(\Delta d)^3 + \abs{S}}{(\Delta d)^4}\right) \le e^{\ell + \abs{S}} \le O(1)^{\abs{S}}$ because $\ell \le \abs{S}$.
   This proves the lemma. \qed

  We now prove \Cref{claim:xor-high-deg-monomials,lemma:disjoint-decomp}.

  \begin{proof} [Proof of \Cref{claim:xor-high-deg-monomials}]
    We may assume $t_j \ge 2$ as otherwise the conclusion is trivial.
    Since ${i_j} \in V_{t_j}(p)$, recalling \Cref{eq:tomato,eq:pepper}, every monomial in $p_j$ containing $x_{i_j}$ has degree at least $t_j$, and by collecting these monomials, we can write
    \begin{align*}
      s_j(x_{i_j}, x_{T_j}, x_{U_j})
      &= x_{i_j} + p_j(x_{i_j},x_{T_j}) + \sum_{k \in S \cap T_j} x_k + q_j(x_{T_j}, x_{U_j}) \\
      &= x_{i_j} (1 + u_j(x_{T_j})) + v_j(x_{T_j}, x_{U_j}) 
    \end{align*}
    for some polynomials $u_j$ and $v_j$, where every monomial in $u_j(x_{T_j})$ has degree at least $t_j-1 \geq 1$, and thus $u_j(0) = 0$.
    Now, if an outcome of $x_{T_j}$ is such that $u_j(x_{T_j}) = 0$, then the expectation of $(-1)^{s_j}$ is zero because $v_j$ does not depend on $x_{i_j}$ and $\E[(-1)^{x_{i_j}}] = 0$.
    Hence we have that for every outcome of $x_{U_j}$, 
    \[
      \abs[\bigg]{\E_{x_{i_j}, x_{T_j}} \Bigl[ (-1)^{s_j(x_{i_j}, x_{T_j}, x_{U_j})} \Bigr]}
      \leq \E_{x_{T_j}} \abs[\bigg]{\E_{x_{i_j}} \Bigl[ (-1)^{s_j(x_{i_j},x_{T_j}, x_{U_j})} \Bigr] }
      \le \Pr[u_j(x_{T_j}) = 1] 
      \le \Delta 2^{-(t_j-1)} ,
    \]
  where the final inequality is by \Cref{lemma:helper}.
  \end{proof}
  
  \begin{lemma} \label{lemma:disjoint-decomp}
    Let $q(x_1, \ldots, x_n)$ be a read-$\Delta$ degree-$d$ polynomial.
    For every $t \in [d]$ we have
    \[
      \abs[\Big]{ \E\bigl[(-1)^q\bigr] }
      \le \exp\left(- \frac{2^{-t} \cdot \abs{V_t(q)}}{(\Delta d)^4} \right) .
    \]
  \end{lemma}

  \begin{proof}
    We partition $V_t(q)$ according to $G_q$ using the following greedy procedure:
    \begin{enumerate}
      \item Set $W_1 = [n]$ and $i=1$.
      \item Pick a monomial $x^{S_i}$ of degree exactly $t$ containing some variable $x_j: j \in W_i \cap V_t(q)$.  
      \item Let $T_i := N_{=1}(S_i)$, and set $W_{i+1} = W_i \setminus N_{\le 3}(S_i)$.
      \item Repeat steps~2 and 3 until we cannot pick such an $S_{\ell+1}$.
      Let $R = [n] \setminus \bigcup_{i \in [\ell]} (S_i \cup T_i)$ be the remaining variables.
    \end{enumerate}
    By construction, the pairwise distances between the $(S_i \cup T_i)$'s are at least $2$, so the variables $x_{S_i \cup T_i}$'s must only appear in disjoint monomials.
    Since each time we remove at most $t \cdot (\Delta d)^3$ elements from $V_t(q)$, we have $\ell \ge \abs{V_t(q)}/(t (\Delta d)^3)$.

    We now partition the monomials in $q$.
    For each $i \in [\ell]$, let $p_i, q_i$ be two polynomials such that $p_i$ is composed of all monomials in $q$ that contain a variable in $x_{S_i}$, and $q_i$ is composed of all monomials in $q$ that contain a variable in $x_{T_i}$ but none in $x_{S_i}$.
    Let $r$ be the polynomial composed of the remaining monomials in $q$.
    We can write $q(x)$ as
    \[
      \sum_{i=1}^\ell \bigl( p_i(x_{S_i}, x_{T_i}) + q_i(x_{T_i}, x_R) \bigr)
      + r(x_R) .
    \]
    By collecting the monomials containing $x^{S_i}$, we can write
    \[
      p_i(x_{S_i}, x_{T_i}) + q_i(x_{T_i}, x_R)
      = x^{S_i} (1 + u_i(x_{T_i})) + v_i(x_{S_i}, x_{T_i}, x_R).
    \]
    for some polynomials $u_i$ and $v_i$, where (1) $u_i$ consists of the at most $\Delta-1$ monomials in $p_i$ that contain $x^{S_i}$, and $u_i(0) = 0$, and (2) $x^{S_i}$ does not appear in any monomial in $v_i$.
    Let $\delta := \Pr[u_i(x_{T_i}) = 0]$.
    By \Cref{lemma:helper} we have $\delta \ge 1/\Delta$.
    Therefore, for every assignment to $x_R$ and every $i \in [\ell]$, 
    \begin{align*}
      \MoveEqLeft
      \E_{x_{T_i}} \abs[\bigg]{\E_{x_{S_i}} \Bigl[ (-1)^{p_i(x_{S_i}, x_{T_i}) + q_i(x_{T_i}, x_R)} \Bigr]} \\
      &= (1-\delta) \abs[\bigg]{\E_{x_{S_i}} \Bigl[(-1)^{v_i(x_{S_i}, x_{T_i}, x_R)} \Bigr]} + \delta \E_{x_{T_i} : u_i(x_{T_i})=0} \left[\abs[\bigg]{\E_{x_{S_i}} \Bigl[(-1)^{x^{S_i} + v_i(x_{S_i}, x_{T_i}, x_R)} \Bigr]}\right] \\
      &\le (1-\delta) + \delta (1 - 2^{-t}) \\
      &\le 1 - 2^{-t}/\Delta ,
    \end{align*}
    where the first inequality is because for any choice of $x_R$ and $x_{T_i}$, the polynomial $x^{S_i} + v_i(x_{S_i}, x_{T_i}, x_R)$ has degree exactly $t$, and therefore its bias is at most $1 - 2^{-t}$.
    Also, for every fixed outcome of $x_R$, the restricted polynomials $p_i(x_{S_i},x_{T_i}) + q_i(x_{T_i},x_R)$ depend on disjoint variables.
    Therefore, 
    \[
      \abs[\bigg]{\E_{x_{T_i},x_{S_i}: i \in [\ell]} \Bigl[(-1)^{q(x)} \Bigr]}
      \le \left(1 - \frac{2^{-t}}{\Delta} \right)^\ell
      \le \exp\left(- \frac{\ell \cdot 2^{-t}}{\Delta} \right)
      \le \exp\left(- \frac{2^{-t} \cdot \abs{V_t(q)}}{(\Delta d)^4}\right) ,
    \]
    because $\ell \ge \abs{V_t(q)}/(t \Delta^3 d^3) \ge \abs{V_t(q)}/(\Delta^3 d^4)$.
\end{proof}


\section{$L_{1,k}$ bounds for disjoint compositions} \label{sec:compositions}

  In this section we give $L_{1,k}$ bounds on \emph{disjoint compositions} of functions, which we define below.

  Let $\calF$ and $\calG$ be two families consisting of functions mapping $\pmo^m$ and $\pmo^\ell$ to $\pmo$ respectively.
  Let $f \in \calF$ and $g_1, \ldots, g_m \in \calG$.
  We define $h\colon\pmo^{m\ell} \to \pmone$, the \emph{disjoint composition} of $f$ and $g_1, \ldots, g_m$, to be
  \[
    h(x_{1,1}, \ldots, x_{1,\ell}, \ldots , x_{m,1}, \ldots, x_{m,\ell})
    := f(g_1(x_{1,1}, \ldots, x_{1,\ell}), \ldots, g_m(x_{m,1}, \ldots, x_{m,\ell})) .
  \]

  \begin{theorem} [Sharper version of \Cref{thm:composition}]
    \label{thm:L1-of-composition}
    Let $g_1, \ldots, g_m \in \calG$, and $f \in \calF$, where $\calG, \calF$ are as above and $\calF$ is closed under restrictions.
    Suppose for every $1 \le k \le K$,
    \begin{enumerate}
      \item $L_{1,k}(f) \le \frac{1-\abs{\E[f]}}{2} \cdot a_\outer \cdot b_\outer^k$ for every $f \in \calF$, and
      \item $L_{1,k}(g) \le \frac{1-\abs{\E[g]}}{2} \cdot a_\inner \cdot b_\inner^k$ for every $g \in \calG$.
    \end{enumerate}
    Then
    \[
      L_{1,K}(h)
      \le \frac{1-\abs{\E[h]}}{2} \cdot a_\outer \cdot b_\inner^K \cdot \frac{a_\inner b_\outer}{2} \left(1 + \frac{a_\inner b_\outer}{2}\right)^{K-1} .
    \]
    In particular, when $a_\inner b_\outer \ge 2$ or $K=1$ we have $L_{1,K}[h] \le \frac{1-\abs{\E[h]}}{2} \cdot a_\outer \cdot \frac{(a_\inner b_\inner b_\outer)^K}{2}$.
  \end{theorem}

\paragraph{Proof idea.}
  Before proving \Cref{thm:L1-of-composition}, we briefly describe the main ideas of the proof.
  For a subset $J \subseteq [m]$, let $\partial_Jf$ denote the \emph{$J$-th derivative of $f$}, which can be expressed as
  \[
    \partial_Jf(x_1, \ldots, x_m) := \sum_{T \supseteq J} \hf(T) x^{T \setminus J} .
  \]
  Note that $\hf(J) = \partial_Jf(\vec{0})$.
  
  Let us begin by considering the task of bounding $L_{1,1}(h) = \sum_{(i,j) \in [m] \times [\ell]} \abs{\hh\{(i,j)\}}$.
  Let $\beta = (\beta_1, \ldots, \beta_m)$, where $\beta_i := \E[g_i]$.
  Using the Fourier expansion of $f$, we have
  \[
    \hh\{(i,j)\}
        = \sum_{S \subseteq [m]} \hf(S) \E \left[\prod_{k \in S} g_k(x_k) \cdot x_{i,j} \right]  .
  \]
  If $S \not\ni i$, then the expectation is zero, because $\prod_{k\in S} g_k(x_k)$ and $x_{i,j}$ are independent and $\E[x_{i,j}] = 0$.
  So, we have
  \[
    \hh\{(i,j)\}
    = \sum_{S \ni i} \hf(S) \beta^{S \setminus \{i\}} \cdot \hg_i(\{j\})
    = \partial_i f(\beta) \cdot \hg_i(\{j\}) .
  \]
  If the functions $g_i$ are balanced, i.e.\ $\E[g_i] = 0$ for all $i$, then we would have $\beta = \vec{0}$, and
  \[
    \hh\{(i,j)\}
    = \partial_i f(\vec{0}) \cdot \hg_i(\{j\})
    = \hf(\{i\}) \hg_i(\{j\}) .
  \]
  So in this case we have
  \[
    L_{1,1}(h)
    = \sum_{i \in [m], j\in [\ell]} \abs[\big]{\hh(\{(i,j)\})}
    = \sum_{i \in [m]} \sum_{j \in [\ell]} \abs{\hf(\{i\})\hg_i(\{j\})}
    = \sum_{i \in [m]} \abs{\hf(\{i\})} \sum_{j \in [\ell]} \abs{\hg_i(\{j\})}
  \]
  and we can apply our bounds on $L_{1,1}(\calF)$ and $L_{1,1}(\calG)$ to $\sum_{i \in [m]} \hf\{i\}$ and $\sum_{j \in [\ell]} \hg_i\{j\}$ respectively.
  Specializing to the case $g_1 = \cdots =g_m$, we have
  \begin{claim} \label{claim:compose-balanced}
    Suppose $g_1 = g_2 = \cdots = g_m =: g$ and $\E[g] = 0$.
    Then $L_{1,1}(h) = L_{1,1}(f) L_{1,1}(g)$.
  \end{claim}
In general the $g_i$'s may not all be the same and may not be balanced, and so it seems unclear how we can apply our $L_{1,1}(\calF)$ bound on $\sum_{i \in [m]} \partial_if(\beta_1, \ldots, \beta_m)$ when $\beta \ne \vec{0}$.
  To deal with this, in \Cref{claim:derivatives} below we apply a clever idea introduced in \cite{CHHL} that lets us relate $f(\beta)$ at a nonzero point $\beta$ to the average of $f_{R_{\beta}}(\vec{0})$, where $f_{R_{\beta}}$ is $f$ with some of its inputs fixed by a random restriction $R_{\beta}$.
  As $\calF$ is closed under restrictions, we have that $f_{R_{\beta}} \in \calF$ and we can apply the $L_{1,1}(\calF)$ bound on $\sum_i \partial_i f_{R_{\beta}}(\vec{0})$, which in turn gives a bound on $\sum_{i\in [m]} \partial_if(\beta_1, \ldots, \beta_m)$.

  Bounding $L_{1,K}(h)$ for $K \ge 2$ is more complicated, as now each $\hh(S)$ involves many $\hf(J)$ and $\hg_i(T)$'s, where the sets $J$ and $T$ have different sizes.
  So one has to group the coefficients carefully.

\subsection{Useful notation}
  For a set $S \subseteq [m] \times [\ell]$, let $S|_f := \{ i \in [m]: (i,j)\in S\text{ for some }j\in [\ell] \}$ be the ``set of first coordinates'' that occur in $S$, and let $S|_i := \{ j \in [\ell] : (i,j) \in S \}$.
  Note that if $(i,j) \in S$, then $i \in S|_f$ and $j \in S|_i$.
  Let $\beta$ denote the vector $(\beta_1, \ldots, \beta_m)$, where $\beta_i := \E[g_i]$ for each $i \in [m]$.
  For a set $J = \{i_1, \ldots, i_{\abs{J}}\} \subseteq [m]$ and $f=f(y_1,\dots,y_m)$, we write $\partial_Jf$ to denote $\frac{\partial^{\abs{J}} f}{\partial y_{i_1} \cdots \partial y_{i_{\abs{J}}}}$.
  Since $\partial_J y^T = \Id(T \supseteq J) y^{T \setminus J}$, by the multilinearity of $f$ we have that
  \begin{align} \label{eq:derivative-in-fourier}
    \partial_Jf(\beta) = \sum_{T \supseteq J} \hf(T) \beta^{\,T \setminus J} .
  \end{align}

\subsection{The random restriction $R_{\beta}$}
  Given $\beta \in [-1,1]^m$, let $R_{\beta}$ be  the random restriction which is the randomized function from $\pmone^m$ to $\pmone^m$ whose $i$-th coordinate is (independently) defined by
  \[
    R_{\beta}(y)_i
    := \begin{cases} \sgn(\beta_i) & \text{with probability $\abs{\beta_i}$} \\
    y_i & \text{with probability $1 - \abs{\beta_i}$.} \end{cases}
  \]
  Note that we have 
  \[
    \E_{R_{\beta}, y}[R_{\beta}(y)_i] = \E_{R_{\beta}}[R_{\beta}(\vec{0})_i] = \beta_i .
  \]
  Define $f_{R_{\beta}}(y)$ to be the (randomized) function $f(R_{\beta}(y))$.
  By the multilinearity of $f$ and independence of the $R_{\beta}(y)_i$ we have
  \[
    \E_{R_{\beta},y}[ f_{R_{\beta}}(y) ]
    = \E_{R_{\beta}}[ f_{R_{\beta}}(\vec{0}) ]
    = f(\beta) .
  \]
  The following claim relates the two derivatives $\partial_Sf(\beta)$ and $\partial_Sf_{R_{\beta}}(\vec{0}) = \widehat{f_{R_{\beta}}}(S)$.

  \begin{claim} \label{claim:derivatives}
    \[
      \partial_{S}f(\beta)
      = \Biggl( \prod_{i \in S} \frac{1}{1 - \abs{\beta_i}} \Biggr) \cdot \E_{R_{\beta}}[ \partial_Sf_{R_{\beta}}(\vec{0}) ]
      = \Biggl( \prod_{i \in S} \frac{1}{1 - \abs{\beta_i}} \Biggr) \cdot \E_{R_{\beta}}[ \widehat{f_{R_{\beta}}}(S) ] .
    \]
  \end{claim}
  \begin{proof}
    The second equality follows from \cref{eq:derivative-in-fourier}.
    To prove the first one, it suffices to show
    \[
      \partial_{\{i\}}f(\beta)
      = \frac{1}{1-\abs{\beta_i}} \E_{R_{\beta}} [ \partial_{\{i\}}f_{R_{\beta}}(\vec{0}) ] .
    \]
    The rest follows by a straightforward induction.
    By definition of derivatives, and noting $\E_{R_\beta}[R_\beta(\frac{1}{1-\abs{\beta_i}} z \cdot e_i)_j]$ equal $\beta_i + z\cdot e_i$ if $i=j$ and $\beta_j$ otherwise,
    \begin{align*}
      \frac{\partial f}{\partial y_i}(\beta)
      &= \lim_{z\to 0} \frac{f(\beta + z \cdot e_i) - f(\beta)}{z} \\
      &= \E_{R_{\beta}} \left[ \lim_{z\to 0} \frac{f_{R_{\beta}} \bigl( \frac{1}{1-\abs{\beta_i}} \cdot z \cdot e_i \bigr) - f_{R_{\beta}}(\vec{0})}{z} \right] \\
      &= \frac{1}{1-\abs{\beta_i}} \E_{R_{\beta}} \left[ \frac{\partial f_{R_{\beta}}}{\partial y_i}(\vec{0}) \right] . \qedhere
    \end{align*}
  \end{proof}
  We now use \Cref{claim:derivatives} to express each coefficient of $h$ in terms of the coefficients of $f$ and $g_i$.
  \begin{lemma} \label{lemma:coeff-in-f-and-g}
  For $S \subseteq [m] \times [\ell]$, we have
  \[
    \hh(S) = \prod_{i \in S|_f} \hg_i(S|_i) \cdot \Biggl( \prod_{i \in S|_f} \frac{1}{1-\abs{\beta_i}} \Biggr) \cdot  \E_{R_{\beta}} \bigl[ \widehat{f_{R_{\beta}}}(S|_f) \bigr].
  \]
  \end{lemma}
  \begin{proof}
    For $S \subseteq [m] \times [\ell]$ and $T \subseteq [m]$, decompose $T$ into $(T \cap S|_f)\,\cup\,(T\setminus S|_f)$ and $S|_f$ into $(T \cap S|_f)\,\cup\,(S|_f\setminus T)$.
    Expanding $h$ using the Fourier expansion of $f$, by the definitions of $S|_f$ and $S|_i$, we have
    \begin{align*}
      \hh(S)
      &= \sum_{T \subseteq [m]} \hf(T) \E_x \left[ \prod_{i\in T} g_i(x) \cdot x^S \right] \\
      &= \sum_{T \subseteq [m]} \hf(T) \left( \Biggl( \prod_{i \in S|_f\setminus T} \E \Bigl[ \prod_{j \in S|_i} x_{i,j} \Bigr] \Biggr) \cdot \Biggl(\prod_{i \in T \cap S|_f} \E \Bigl[ g_i(x) \prod_{j \in S|_i} x_{i,j} \Bigr] \Biggr) \cdot \Biggl( \prod_{i \in T\setminus S|_f} \E[g_i(x)] \Biggr) \right) .
    \end{align*}
    When $S|_f \not\subseteq T$, we have $S|_f\setminus T \ne \varnothing$ and thus $\prod_{i \in S|_f\setminus T} \prod_{j \in S|_i} \E[ x_{i,j} ] = 0$.
    Hence,
    \[
      \prod_{i \in S|_f\setminus T} \E \Bigl[ \prod_{j \in S|_i} x_{i,j} \Bigr]
      =\Id(S|_f \subseteq T) .
    \]
    Note that $\E[ g_i(x) \prod_{j \in S} x_{i,j}]  = \hg_i(S)$ for every $S \subseteq [\ell]$.
    Using $\E[g_i(x)] = \beta_i$, \Cref{eq:derivative-in-fourier} and \Cref{claim:derivatives}, we have
    \begin{align*}
      \hh(S)
      &= \sum_{T \subseteq [m]: T \supseteq S|_f} \hf(T) \prod_{i \in S|_f} \hg_i(S|_i) \prod_{i \in T\setminus S|_f} \E[g_i(x)] \\
      &= \prod_{i \in S|_f} \hg_i(S|_i) \cdot \Biggl( \sum_{T \subseteq [m]: T \supseteq S|_f} \hf(T) \prod_{i \in T\setminus S|_f} \beta_i \Biggr) && \text{($\E[g_i(x)] = \beta_i$)}\\
      &= \prod_{i \in S|_f} \hg_i(S|_i) \cdot \partial_{S|_f}f(\beta_1, \ldots, \beta_m) && 
     \text{(\Cref{eq:derivative-in-fourier})} \\
      &= \prod_{i \in S|_f} \hg_i(S|_i) \cdot \Biggl( \prod_{i \in S|_f} \frac{1}{1-\abs{\beta_i}} \Biggr) \cdot  \E_{R_{\beta}} \bigl[ \widehat{f_{R_{\beta}}}(S|_f) \bigr] .&&  \text{(\Cref{claim:derivatives})} \qedhere
    \end{align*}
  \end{proof}

\subsection{Proof of \Cref{thm:L1-of-composition}}
    By \Cref{lemma:coeff-in-f-and-g}, $L_{1,K}(h)$ is equal to
    \[
      \sum_{S \subseteq [m] \times [\ell] : \abs{S}=K} \abs{\hh(S)} 
      = \sum_{S \subseteq [m] \times [\ell] : \abs{S}=K} \Biggl| \Biggl( \prod_{i \in S|_f} \hg_i(S|_i) \Biggr) \cdot \Biggl( \prod_{i \in S|_f} \frac{1}{1-\abs{\beta_i}} \Biggr) \cdot  \E_{R_{\beta}} \bigl[ \widehat{f_{R_{\beta}}}(S|_f) \bigr] \Biggr| .
    \]
    We enumerate all the subsets $S \subseteq [m] \times [\ell]$ of size $K$ in the following order: 
    For every $\abs{J} = k \in [K]$ out of the $m$ blocks of $\ell$ coordinates, we enumerate all possible combinations of the (disjoint) nonempty subsets $\{ S_i: i \in J\}$ in those $k$ blocks whose sizes sum to $K$.
    Rewriting the summation above in this order, we obtain
    \begin{align}
      \MoveEqLeft
      \sum_{S \subseteq [m] \times [\ell] : \abs{S}=K} \abs{\hh(S)} \\
      &= \sum_{k=1}^K \sum_{\substack{J \subseteq [m] \\ \abs{J}=k}} \sum_{\substack{w \subseteq [\ell]^J \\ \sum_{i \in J} w_i = K}} \sum_{\substack{\{S_i\}_{i \in J} \subseteq [\ell]^J : \\ \forall i \in J: \abs{S_i} = w_i}} \abs[\Bigg]{ \Biggl( \prod_{i\in J} \hg_i(S_i) \Biggr) \cdot \Biggl( \prod_{i\in J} \frac{1}{1-\abs{\beta_i}} \Biggr) \cdot \E_{R_{\beta}} \bigl[ \widehat{f_{R_{\beta}}}(J) \bigr] } \nonumber \\
      &\le \sum_{k=1}^K \sum_{\substack{J \subseteq [m] \\ \abs{J}=k}} \sum_{\substack{w \subseteq [\ell]^J \\ \sum_{i \in J} w_i = K}} \sum_{\substack{\{S_i\}_{i \in J} \subseteq [\ell]^J : \\ \forall i \in J: \abs{S_i} = w_i}} \Biggl( \prod_{i\in J} \abs[\big]{\hg_i(S_i)}\Biggr) \cdot \Biggl( \prod_{i\in J} \frac{1}{1-\abs{\beta_i}} \Biggr)  \cdot \abs[\Big]{\E_{R_{\beta}}\bigl[  \widehat{f_{R_{\beta}}}(J) \bigr]} \label{eq:stop1}.
    \end{align}
    Since $L_{1,w_i}(g_i) \le \frac{1-\abs{\beta_i}}{2} \cdot a_\inner \cdot b_\inner^{w_i}$, for every $\{w_i\}_{i\in J}$ such that $\sum_{i\in J} w_i = K$, we have
    \[
      \sum_{\substack{\{S_i\}_{i \in J} \subseteq [\ell]^J : \\ \forall i \in J: \abs{S_i} = w_i}} \prod_{i\in J} \abs{\hg_i(S_i)}
      = \prod_{i\in J} L_{1,w_i}(g_i)
      \le \prod_{i \in J} \Bigl( \frac{1-\abs{\beta_i}}{2} a_\inner b_\inner^{w_i} \Bigr)
      = b_\inner^K a_\inner^{\abs{J}} \prod_{i \in J} \frac{1-\abs{\beta_i}}{2}.
    \]
    Plugging the above into \cref{eq:stop1}, we get that
    \begin{align}
      \sum_{S \subseteq [m] \times [\ell] : \abs{S}=K} \abs{\hh(S)} 
      &\le b_\inner^K \sum_{{k}=1}^K  a_\inner^{k} \sum_{\substack{J \subseteq [m] \\ \abs{J}={k}}} \sum_{\substack{w \subseteq [\ell]^J \\ \sum_{i \in J} w_i = K}} \prod_{i\in J} \left(\frac{1-\abs{\beta_i}}{2} \cdot \frac{1}{1-\abs{\beta_i}} \right) \cdot \abs[\Big]{\E_{R_{\beta}} \bigl[\widehat{f_{R_{\beta}}}(J) \bigr]} \nonumber \\
      &= b_\inner^K \sum_{{k}=1}^K  \left(\frac{a_\inner}{2}\right)^k \sum_{\substack{J \subseteq [m] \\ \abs{J}={k}}} \abs[\Big]{\E_{R_{\beta}} \bigl[\widehat{f_{R_{\beta}}}(J) \bigr]} \sum_{\substack{w \subseteq [\ell]^J \\ \sum_{i \in J} w_i = K}} 1  \nonumber \\
      &\le b_\inner^K \sum_{{k}=1}^K \left(\frac{a_\inner}{2}\right)^k \binom{K - 1}{k - 1}  \sum_{\substack{J \subseteq [m] \\ \abs{J}={k}}} \abs[\Big]{\E_{R_{\beta}} \bigl[\widehat{f_{R_{\beta}}}(J) \bigr]} \label{eq:stop2} ,
    \end{align}
    where the last inequality is because for every subset $J \subseteq [m]$, the set
     $\{w \subseteq [\ell]^J: \sum_{i \in J} w_i = K\}$ has size at most $\binom{K-1}{\abs{J}-1}$.
    We now bound $\abs{\E_{R_{\beta}}[\widehat{f_{R_{\beta}}}(J)]}$.
    Since for every restriction $R_{\beta}$, we have $f_{R_{\beta}} \in \calF$ (by the assumption that $\calF$ is closed under restrictions), it follows that
    \[
      L_{1,{k}}(f_{R_{\beta}})
      \le \frac{1-\abs{\E_y[f_{R_{\beta}}(y)]}}{2} a_\outer b_\outer^k
      \le \frac{1-\E_y[f_{R_{\beta}}(y)]}{2} a_\outer b_\outer^{k} .
    \]
    So
    \begin{align*}
      \sum_{J \subseteq [m], \abs{J}={k}} \abs[\Big]{\E_{R_{\beta}} \bigl[\widehat{f_{R_{\beta}}}(J) \bigr]}
      &\le \E_{R_{\beta}}[ L_{1,{k}}(f_{R_{\beta}}) ] \\
      &\le \frac{1-\E_{R_{\beta},y}[f_{R_{\beta}}(y)]}{2} a_\outer b_\outer^k \\
      &= \frac{1-\E[h]}{2} a_\outer b_\outer^k .
    \end{align*}
    Continuing from \cref{eq:stop2}, we get
    \begin{align*}
        \sum_{S \subseteq [m] \times [\ell] : \abs{S}=K} \abs{\hh(S)}
        &\le \frac{1-\E[h]}{2} \cdot b_\inner^K \cdot \sum_{k=1}^K  \left(\frac{a_\inner}{2}\right)^k \cdot \binom{K-1}{k-1} \cdot a_\outer b_\outer^k \\
        &= \frac{1-\E[h]}{2} \cdot a_\outer \cdot b_\inner^K \cdot \frac{a_\inner b_\outer}{2} \left(1 + \frac{a_\inner b_\outer}{2}\right)^{K-1}
    \end{align*}
    where the last step used the binomial theorem.
    Applying the same argument to $-h$ lets us replace $\frac{1-\E[h]}{2}$ with $\frac{1-\abs{\E[h]}}{2}$, concluding the proof of \Cref{thm:L1-of-composition}.
\qed

\subsection{Bounding $L_{1,1}(p)$ for balanced polynomials}
  We give a quick application of \Cref{claim:compose-balanced} to show that $L_{1,1}(p) \le d$ for balanced degree-$d$ polynomials.
  (Note that \cite{CHLT} showed that $L_{1,1}(p) \le 4d$ for general degree-$d$ polynomials.)

  \begin{claim}
    For every degree-$d$ polynomial $p\colon\zo^n \to \zo$ with $\E[p]=1/2$, we have $L_{1,1}(p) \le d$.
  \end{claim}
  \begin{proof}
    Suppose not, and suppose that there exists an $\eps > 0$ such that $L_{1,1}(p) = (1+\eps) d$ for some degree-$d$ balanced polynomial $p$ on $n$ variables.
    Consider the composition $p'$ on disjoint sets of variables with the inner functions $g_1 = g_2 = \dots g_n$ and the outer function all equal to $p$.
    Note that $p'$ is a balanced degree $d^2$ polynomial with $L_{1,1}(p') = (1+\eps)^2 d^2$.

    We can repeat this iteratively by replacing $p$ in $p'$ in the composition.
    After $k$ repetitions, we get a polynomial $q$ of degree $d^{2^k}$ that has $L_{1,1}(q) = (1+\eps)^{2^k} d^{2^k}$. For large enough $k$ this is greater than $4d^{2^k}$, but this contradicts the $4d$ upper bound of \cite{CHLT}, a contradiction.
  \end{proof}

\subsection*{Acknowledgements}
We thank Shivam Nadimpalli for stimulating discussions at the early stage of the project, and Pooya Hatami for pointing out a mistake in \Cref{lemma:read-few-key} in a previous version of the paper.
We also thank the anonymous reviewers for their helpful comments.

Jaros\l{}aw B\l{}asiok is supported by a Junior Fellowship from the Simons Society of Fellows.
Peter Ivanov and Emanuele Viola are supported by NSF grant CCF-2114116.
Yaonan Jin is supported by NSF grants IIS-1838154, CCF-1563155, CCF-1703925, CCF-1814873, CCF-2106429, and CCF-210718.
Chin Ho Lee is supported by the Croucher Foundation and by the Simons Collaboration on Algorithms and Geometry while at Columbia University, and by Simons Investigator grants to Salil Vadhan and Madhu Sudan.
Rocco Servedio is supported in part by NSF awards CCF-1563155, IIS-1838154 and CCF-2106429, and by the Simons Collaboration on Algorithms and Geometry.

\begin{flushleft}
\bibliographystyle{alpha}
\bibliography{ref}
\end{flushleft}

\appendix


\section{Reduction to bound without acceptance probability} \label{sec:reduction}
In this section, we show that given any $L_{1,k}$ Fourier norm bound on a class of functions that is closed under XOR on disjoint variables, such a bound can be automatically ``upgraded'' to a refined bound that depends on the acceptance probability:

\begin{lemma} \label{lemma:xor-reduction-restate}
  Let $\calF$ be a class of $\pmo$-valued functions such that for every $f\in \calF$, the XOR of disjoint copies of $f$ (over disjoint sets of variables) also belongs to $\calF$.
  If $L_{1,k}(\calF) \le b^k$, then for every $f \in \calF$ it holds that $L_{1,k}(f) \le 2e \cdot \frac{1-\abs{\E[f]}}{2} \cdot b^k$.
\end{lemma}
\begin{proof}
  Suppose not, and let $f \in \calF$ be such that $L_{1,k}(f) > 2e \cdot \frac{1-\abs{\E[f]}}{2} \cdot b^k$.
  We first observe that since $L_{1,k}(\calF) \leq b^k$, it must be the case that $1-\abs{\E[f]} \leq 1/e.$
  Let $\alpha := \frac{1-\abs{\E[f]}}{2} \in [0,\frac{1}{{2e}}]$ so that $\abs{\E[f]} = 1-2\alpha \geq 1 - 1/e$.
  Let $f^{\oplus t}$ be the XOR of $t$ disjoint copies of $f$ on $tn$ variables, where the integer $t$ is to be determined below.
  By our assumption, we have $f^{\oplus t} \in \calF$ and thus
  \begin{align*}
      L_{1,k}(f^{\oplus t})
      &\ge \binom{t}{1} \cdot L_{1,0}(f)^{t-1} \cdot L_{1,k}(f)
      \tag{by disjointness} \\
      &= t \cdot (1-2\alpha)^{t-1} \cdot L_{1,k}(f)
      \tag{$L_{1,0}(f)=\E[f]$} \\
      &> t \cdot (1-2\alpha)^{t-1} \cdot 2e \cdot \alpha \cdot b^k =: \Lambda(t).  
  \end{align*}
  We note that if $\alpha=0$ then $\abs{\E[f]}=1$, so all the Fourier weight of $f$ is on the constant coefficient, and hence the claimed inequality holds trivially.  So we subsequently assume that $0<\alpha\leq \frac{1}{2e}$.
  Let $t^* := \frac{1}{-\ln(1 - 2\alpha)} > 0$.
  It is easy to verify that $\Lambda(t)$ is increasing when $t \leq t^*$, and is decreasing when $t \geq t^*$.
  
  We choose $t = \lceil t^* \rceil$.\ignore{ and finish the proof by a case analysis.}
  Since $\alpha \leq \frac{1}{2e} < \frac{e - 1}{2e} \approx 0.3161$, we have $t^* > 1$ and thus
  \begin{align*}
      L_{1,k}(f^{\oplus t})
      > \Lambda(\ceil{t^*}) 
      \ge \Lambda(t^* + 1) 
      &= \left(\frac{1}{-\ln(1 - 2\alpha)} + 1\right) \cdot (1-2\alpha)^{\frac{1}{-\ln(1 - 2\alpha)}} \cdot 2e \cdot \alpha \cdot b^k \\
      &= \left(\frac{2\alpha}{-\ln(1 - 2\alpha)} + 2\alpha\right) \cdot b^k 
      \ge b^k,
  \end{align*}
  where the last inequality holds for every $\alpha \in (0, \frac{e - 1}{2e}]$ and can be checked via elementary calculations.
  \ignore{
  In the other case where $\alpha \geq \frac{e - 1}{2e}$, we have $t^* \leq 1$ and then $t = 1$. As a consequence,
  \begin{align*}
      L_{1,k}(f^{\oplus t})
      > \Lambda(1) 
      = 2e \cdot \alpha \cdot b^k
      \geq (e - 1) \cdot b^k
      > b^k.
  \end{align*}
  To conclude, we always get a contradiction that $L_{1,k}(f^{\oplus t}) > b^k$ for some function $f^{\oplus t} \in \calF$.
  }This contradicts $L_{1,k}(\calF) \leq b^k$, and the lemma is proved.
\end{proof}

\end{document}